\documentclass[11pt]{amsart}
\usepackage[margin=1in]{geometry}

\usepackage{graphicx}
\usepackage{amsfonts,amsmath,amssymb,multirow}
\usepackage{multirow}
\usepackage[usenames, dvipsnames]{color}
\usepackage{algorithm}
\usepackage[noend]{algorithmic}

\newtheorem{theorem}{Theorem}[section]
\newtheorem{proposition}[theorem]{Proposition}
\newtheorem{lemma}[theorem]{Lemma}
\newtheorem{corollary}[theorem]{Corollary}
\newtheorem{definition}[theorem]{Definition}

\newcommand{\bE}{\ensuremath{\mathbf{E}}}

\newcommand{\polylog}{\operatorname{polylog}}
\newcommand{\ID}{\operatorname{ID}}
\newcommand{\poly}{\operatorname{poly}}

\newcommand{\pal}{\operatorname{Pal}}
\newcommand{\extd}{\ensuremath{{\overline{d}}}}

\begin{document}

\author[David G. Harris and Johannes Schneider and Hsin-Hao Su]{
{\sc David G.~Harris}$^{2}$
\and
{\sc Johannes Schneider}$^{3}$
\and
{\sc Hsin-Hao Su}$^{4}$
}

\setcounter{footnote}{0}

\addtocounter{footnote}{1}
\footnotetext{This is an extended version of a paper which appeared in the 2016 Symposium on Theory of Computing (STOC)}

\addtocounter{footnote}{1}
\footnotetext{Department of Computer Science, University of Maryland,
College Park, MD 20742.
Research supported in part by NSF Awards CNS-1010789 and CCF-1422569.
Email: \texttt{davidgharris29@gmail.com}.}

\addtocounter{footnote}{1}
\footnotetext{University of Liechtenstein, Vaduz, Liechtenstein. Email: \texttt{johannes.schneider@uni.li}.}

\addtocounter{footnote}{1}
\footnotetext{University of North Carolina at Charlotte. Part of the work was done while at MIT and University of Michigan, supported by NSF Awards CCF-0939370, CCF-1217338, BIO-1455983, and AFOSR FA9550-13-1-0042. Email: \texttt{hsinhaosu@uncc.edu}.}

\title{Distributed $(\Delta+1)$-Coloring in Sublogarithmic Rounds$^{1}$}

%\author{name \\ organisation\\ address\\ email    } johannes.schneider@alumni.ethz.ch
\maketitle
\section*{Abstract} \label{sec:Abstract}
\noindent
 We give a new randomized distributed algorithm for $(\Delta+1)$-coloring in the LOCAL model, running in $O(\sqrt{\log \Delta})+ 2^{O(\sqrt{\log \log n})}$ rounds in a graph of maximum degree~$\Delta$. This implies that the $(\Delta+1)$-coloring problem is easier than the maximal independent set problem and the maximal matching problem, due to their lower bounds of $\Omega \left( \min \left( \sqrt{\frac{\log n}{\log \log n}}, \frac{\log \Delta}{\log \log \Delta}  \right)  \right)$ by Kuhn, Moscibroda, and Wattenhofer [PODC'04]. Our algorithm also extends to list-coloring where the palette of each node contains $\Delta+1$ colors. We extend the set of distributed symmetry-breaking techniques by performing a decomposition of graphs into dense and sparse parts.
%Based on a local density measure we decompose a graph into sparse and dense parts, whereas prior work is based on non-local properties [FOCS '89]. Its application yields a new randomized coloring algorithm for $\Delta+1$ coloring running in time~$O(\sqrt{\log \Delta}+ 2^{O(\sqrt{\log \log n})})$ time with probability~$1-1/n^{\Omega(1)}$ in a general graph with $n$ nodes with maximal degree~$\Delta$. Our algorithm improves the state of art up [FOCS '12] up to a factor of $\sqrt{\log \Delta}$. We also cover list-colorings for locally sparse graphs.

\bigskip

\noindent Categories and Subject Descriptors: G.2.2 \textbf{[Graph Theory]}: Graph Algorithms; F.2.2 \textbf{[Nonnumerical Al-
gorithms and Problems]}: Computations on Discrete Structures

\medskip

\noindent General Terms: Algorithms, Theory

\medskip

\noindent Additional Key Words and Phrases: Distributed Networks, Vertex Coloring

\section{Introduction} \label{sec:Intro}
Given a graph $G = (V,E)$, let~$n = |V|$ denote the number of vertices and let~$\Delta$ denote the maximum degree. The $k$-coloring problem is to assign each vertex~$v$ a color~$\chi(v) \in \{1,2,\ldots, k\}$ such that no two neighbors are assigned the same color.

In this paper, we study the $(\Delta+1)$-coloring problem in the distributed $\mathsf{LOCAL}$ model. In this model, vertices host processors and operate in synchronized rounds. In each round, each vertex sends a message of arbitrary size to each of its neighbors, receives messages from its neighbors, and performs (unbounded) local computations. The time complexity of an algorithm is measured by the number of rounds until every vertex commits its output -- in our case, its color. The $\mathsf{LOCAL}$ model is a model for investigating what {\it local information is needed for each vertex to compute its own output}. An $r$-round algorithm in the $\mathsf{LOCAL}$ model implies that each vertex only uses information in its $r$-neighborhood to compute the output, and vice versa. 

Graph coloring is one of the central problems of graph theory, with numerous applications to algorithms and combinatorics. The  $(\Delta+1)$-coloring problem is the most celebrated case, because for every graph there exists a $(\Delta+1)$ coloring, which can be found via a simple (sequential) greedy algorithm. Furthermore, even if $\Delta$-colorings exist for some graphs, there are examples where it cannot be solved locally; for example, 2-coloring a ring requires $\Omega(n)$ rounds. 

Vertex-coloring in the distributed model, in particular, has applications as a subroutine in other distributed algorithms (for example, the Lov\'{a}sz Local Lemma \cite{fischer2017sublogarithmic}) and scheduling problems (for example, radio network broadcasts~\cite{CK85, CS89, RP89, ET90}). It is typical in such applications that solution quality depends on the number of colors used. Thus $(\Delta+1)$-coloring is a natural problem, as it leads to the optimal schedules that one can hope to guarantee as a function of $\Delta$.

\subsection{Previous results on distributed graph coloring}
The distributed coloring problem, and variants, have a long history dating back to the 1980's. Table \ref{ResultsOverview} summarizes the results for the most prominent coloring problem, i.e. $(\Delta+1)$-coloring. Several deterministic algorithms have been developed with a run-time of  $O(f(\Delta) + \log^{*} n)$~\cite{bar15a, BEK09,Kuhn2006On,linial92,GPS88,GP87}.  The latter term is necessary as 3-coloring a ring require~$\Omega(\log^{*} n)$ rounds for both deterministic and randomized algorithms \cite{linial92, Noar91}.  A breakthrough by Barenboim~\cite{bar15a} gave an algorithm running in $O(\Delta^{3/4} \log \Delta + \log^{*} n)$ rounds, subsequently improved to $O(\sqrt{\Delta}\log^{2.5}\Delta + \log^{*} n)$ by Fraigniaud et al.~\cite{FraigniaudHK15}. These bounds have a sublinear dependence on~$\Delta$, which is notable since~$\Omega(\Delta)$ lower bounds hold for related problems in more restrictive settings~\cite{ goo14, hef16, hir12, Kuhn2006On, SV93}.

The randomized algorithms considered in this paper are those which return a correct answer w.h.p.~({\it with high probability}, which is with probability at least~$1 - n^{-K'}$, where~$K'$ is an arbitrary constant~$K'>0$.) Randomized approaches can be traced back to the $O(\log n)$ rounds maximal independent set (MIS) algorithm of Alon, Babai, and Itai~\cite{alon86} and Luby~\cite{lub86}. The $O(\log n)$ upper bound lasted until Schneider and Wattenhofer gave an algorithm of running time~$O(\log \Delta + \sqrt{\log n})$  \cite{Sch10}. Then, Barenboim et al.~\cite{BEPS16} improved the dependence on $n$ to $2^{O(\sqrt{\log \log n})}$. All these algorithms require $\Omega(\log n)$ rounds when $\Delta = n^{c}$ for some constant $0<c \leq 1$. 

\subsection{Our contributions.} We give a randomized algorithm running in time~$O(\sqrt{\log \Delta}) + 2^{O(\sqrt{\log \log n})}$ rounds w.h.p., which is the first algorithm that runs in $o(\log n)$ rounds for every graph. Moreover, this implies a separation between the $(\Delta+1)$-coloring and the MIS problem. We elaborate our contributions in the following:

\begin{table*}
\centering
\begin{tabular}{|l | l r |l r | }    \hline
   Bounds  & Randomized & & Deterministic &\\ \hline
    \multirow{ 8}{*}{Upper} & $O(\sqrt{\log \Delta})+2^{O(\sqrt{\log \log n})}$ & [This paper]  & $O(\sqrt{\Delta}\log^{2.5}\Delta + \log^{*} n)$ & \cite{FraigniaudHK15}  \\ \cline{2-5}% \hline % & ra. & O($\log \chi + \sqrt{\log n} )$ [This paper]\\ \cline{2-3}     %\cline{1-2}
      &$O(\log \Delta)+2^{O(\sqrt{\log \log n})}$ & \cite{BEPS16} & $O(\Delta^{3/4}\log \Delta+ \log^* n)$ & \cite{bar15a}  \\ \cline{2-5}% \hline
     &$O(\log \Delta+\sqrt{\log n})$  &\cite{Sch10}  &  $O(\Delta+ \log^* n)$ & \cite{BEK09}  \\ \cline{2-5}% \hline
    &$O(\log n)$ & \cite{lub86,alon86,Joh99} &   $O(\Delta\log \Delta + \log^* n)$ & \cite{Kuhn2006On} \\ \cline{2-5}% \hline
    & & & $O(\Delta^2 + \log^{*} n)$ & \cite{linial92, GPS88} \\ \cline{4-5}
    & & & $O(\Delta \log n)$ & \cite{GPS88} \\ \cline{4-5}%\hline
    & & & $\Delta^{O(\Delta)} + O(\log^{*} n)$  & \cite{GP87} \\ \cline{4-5}%\hline
    & $ $ &    & $2^{O(\sqrt{\log n})}$ & \cite{panc92}   \\ \hline
    \hline
    Lower& $\Omega(\log^{*} n)$ & \cite{Noar91} & $\Omega(\log^{*} n)$ & \cite{linial92}   \\ \hline
  \end{tabular}
\caption{Comparison of $(\Delta +1)$-coloring algorithms and lower bounds} \label{ResultsOverview}
\vspace{-5mm}
\end{table*}

%Our main result is a new algorithm for the $(\Delta+1)$-coloring problem with a round complexity of $O(\sqrt{\log \Delta}) + 2^{O(\sqrt{\log \log n})}$ in the \textsf{LOCAL} model. The contribution is twofolds, explained as follows.  

\begin{enumerate}
\item {\bf Separation between the coloring problem and the MIS problem.}  The coloring problem and the MIS problem are closely related; for example, given a $(\Delta+1)$-coloring one can compute a MIS in $\Delta+1$ rounds by letting a node with color~$i$ join the MIS in round~$i$ (if no neighbor joined previously). Conversely, Lov\'{a}sz describes how any MIS algorithm can be used for $(\Delta+1)$-coloring in the same running time by simulating it on a blow-up graph \cite{lov79} (this result has also been mentioned in \cite{lub86} and \cite{alon86}). Kuhn, Moscibroda, and Wattenhofer \cite{kuh10} constructed a family of graphs with $\Delta = 2^{O(\sqrt{\log n \log \log n})}$ for which computing an MIS or a maximal matching requires at least~$\Omega(\sqrt{\frac{\log n}{\log \log n}})$ rounds. %\footnote{The same authors claimed the bound can be improved to $\Omega(\sqrt{\log n})$ with graphs  of $\Delta = 2^{O(\sqrt{\log n})}$ in \cite{kuh10}. However, recently
%Bar-Yehuda, Censor-Hillel, and Schwartzman \cite{BCS16} pointed out an error in their proof.}. 
To this date, it has been unclear whether $(\Delta +1)$-coloring,  MIS and maximal matching are equally hard problems.

As our algorithm computes $(\Delta+1)$-colorings in the above graphs in $O((\log n \log \log n)^{1/4})$ rounds, we show that $(\Delta+1)$-coloring is an easier problem.

\item {\bf Breaking the $O(\log n)$ barrier.} From a technical perspective, the {\it union bound barrier} for randomized distributed algorithms was observed in \cite{bar15}. Roughly speaking, randomized algorithms generally conduct a series of trials at each vertex. When the trial at a vertex succeeds for the first time, then it commits its output. If the failure probability of each trial is at least~$p$, then it takes $\Omega(\log_{1/p} (1/\delta) )$ trials to ensure a vertex succeed with probability at least~$1 - \delta$ for~$0 < \delta < 1$.  However, since we require all the nodes to correctly output their answers at the end of the algorithm, $\delta$ has to be bounded by $1/\poly(n)$ in order to take the union bound over every vertex. Therefore, many randomized algorithms require at least~$\Omega(\log_{1/p} n)$ rounds. 
	
Barenboim et al.~\cite{bar15a} used the {\it graph shattering technique} to circumvent the union bound barrier. The basic idea is to run the experiments for~$\Omega(\log_{1/p} \Delta)$ rounds so that the probability each vertex failed is at most~$1/\poly(\Delta)$ after this first phase. The size of each connected component in the graph induced by unsuccessful vertices then becomes $\polylog (n)$. Then, one may run a deterministic algorithm on each component in parallel in the second phase. The running time of Panconesi and Srinivasan's deterministic algorithm is~$2^{O(\sqrt{\log N})}$ on graphs of size $N$~\cite{panc92}. Since the size of the component is exponentially smaller than the original graph, the running time scales down correspondingly to be~$2^{O(\sqrt{\log \log n})}$. 
	
The graph shattering technique does not directly apply to $(\Delta+1)$-coloring, since every algorithm up to this date has had a failure probability per round (and node) lower-bounded by constants. To appreciate the difficulty of achieving sub-constant failure probability, consider the following natural approach to the $(\Delta+1)$-coloring problem: each vertex selects a color randomly and commits to the color if no neighbors has selected the same color. If the graph contains a clique of size $\Delta$, then the probability that a vertex in the clique successfully colors itself is $(1-\frac{1}{\Delta+1})^{\Delta} = \Theta(1)$. %It can shown easily that the failure probability is also lower bounded by a constant in the subsequent rounds. 
(The MIS problem is even harder, and it was only recently shown by \cite{Moh16} how to achieve {\it amortized} constant  failure probability in $O(1)$~rounds.)

To break the $O(\log n)$ barrier, we introduce novel ideas to address symmetry-breaking problems. We develop a network decomposition that splits a graph into sparse and dense parts and tackles them separately. The key is that vertices boost their probability of getting colored by using either the properties of dense parts or those of sparse parts (see Section \ref{sec:technical}). Our algorithm does not directly use graph shattering, although we use a number of previous algorithms as subroutines which do use this technique.
\end{enumerate}

Our algorithm extends to a closely related generalization of the vertex-coloring problem known as \emph{list-coloring}. Here, there is a set of colors~$\mathcal C$, and each vertex is equipped with a palette~$\pal(v) \subseteq \mathcal C$ of size $|\pal(v)| \geq \Delta+1$; each vertex selects one color from its palette, and no two neighbors can be assigned the same color. $(\Delta+1)$-coloring is a special case in which $\mathcal C = \{1, \dots, \Delta + 1 \}$ and every vertex has the same palette of size~$\Delta+1$. 

\subsection{Technical summary}\label{sec:technical}
%\textcolor{red}{Suggestion for first few sentences:  To begin with, if we use more colors than necessary then it is possible to color the graph faster, e.g. for $O(\Delta)$-coloring algorithms consult \cite{Sch10, elk15}. Similar ideas apply to graphs of ``limited'' density. For example, triangle-free ...}

The possibility of coloring a vertex with super-constant probability was first observed in \cite{Sch10, elk15}, in the setting where there are~$(1+\Omega(\epsilon))\Delta$ colors for some~$\epsilon > 0$. The idea is that vertices try for multiple colors on each trial, if the palette size exceeds a node's degree. This yields a~$O(\log (1/\epsilon) + \log^{*} n)$ rounds algorithm for the first phase. Combined with the graph shattering technique, the algorithm runs in $O(\log (1/\epsilon)) + 2^{O(\sqrt{\log \log n})}$ rounds. Elkin, Pettie, and Su~\cite{elk15} observed that if a graph is $(1-\epsilon)$-locally-sparse (each vertex participates in at most $(1-\epsilon)\binom{\Delta}{2}$ triangles), then it can be reduced to the coloring problem with $(1+\Omega(\epsilon^2))\Delta$ colors.

%We begin by observing that if we use more colors than needed, then graph coloring can be executed faster; for example, graphs can be colored very fast using $(1+\Omega(1))\Delta$ colors \cite{Sch10, elk15}. Similar ideas apply to sparse graphs, whose chromatic number is known to be smaller than~$\Delta+1$. Elkin, Pettie, and Su \cite{elk15} showed that if a graph is $(1-\epsilon)$-locally-sparse, then it is possible to obtain a $(\Delta+1)$-coloring in $O(\log (1/\epsilon)) + 2^{O(\sqrt{\log \log n})}$ rounds.

It is thus the dense parts of the graph that become bottlenecks. However, if a subgraph is dense, then it is likely to have small (weak) diameter. (The {\it weak diameter} of a subgraph $H$ is the maximum distance measured in $G$ between any pair of vertices $u,v \in H$.) A single vertex in $H$ can read in all the information in~$H$, make a decision, and broadcast it to $H$ in time proportional to its weak diameter.

We develop a network decomposition procedure based on \emph{local sparseness}. Our decomposition algorithm is targeted towards identifying dense components of constant weak diameter and sparse components in a constant number of rounds. Roughly speaking, a sparse vertex is one which participates in at most $(1-\epsilon) \binom{\Delta}{2}$ triangles in its neighborhood, where $\epsilon > 0$ is a parameter that we will carefully choose. At the same time, we would also like to bound the number of neighbors of a dense component that are not members of the dense component itself, called \emph{external neighbors}. This step is necessary to bound the influence of color choices of nodes in one component on other components. This mechanism may help to leverage algorithms for other distributed problems that can handle either dense or sparse graphs well.

First, we ignore the sparse vertices. Since each dense component has constant weak diameter, it can elect a leader to assign a color to every member so that no intra-component conflicts occur, i.e.~the endpoints of the edges inside the same component are always assigned different colors. Meanwhile, we hope that the assignments are random enough so that the chance of inter-component conflicts will be small.
Combined with the property of the decomposition that the number of external neighbors is bounded, we show that the probability that a vertex remains uncolored is roughly $O(\epsilon)$ in each round. After $O(\log_{1/\epsilon} \Delta)$ rounds, the degree of each vertex becomes sufficiently small so that the algorithm of Barenboim et al.~\cite{bar15a} can handle the residual graph efficiently.

For the sparse vertices we analyze a preprocessing {\it initial coloring step} of the algorithm. In a similar vein as \cite{elk15}, we show that there will be an $\Omega(\epsilon^2 \Delta)$ gap between the palette size and the degree due to the sparsity. The gap remains while the dense vertices are colored. So, we will be able to color the sparse vertices by using the algorithm of Elkin et al.~\cite{elk15}, which requires $O(\log (1/\epsilon)) + \exp(O(\sqrt{\log \log n}))$ rounds. In contrast to \cite{elk15}, our analysis generalizes to the list-coloring problem. By setting $\epsilon = 2^{-\Theta(\sqrt{\log \Delta})}$, we balance the round complexity between the dense part and the sparse part, yielding the desired running time.

The main technical challenge lies in the dense components. In each component, we need to generate a random proper coloring so that each vertex has a small probability of receiving the same color as one of its external neighbors.
%A tempting approach is to find any proper coloring first, and then randomly permute the colors. %However, while this may work when the palettes are all identical, after the second round, the available %colors for different vertices may be different and this approach may fail.
 We give a process for generating a proper coloring where the probability that a vertex gets any color from its palette is close to uniform. Additionally, we need to show that the structure of the decomposition is maintained in the next round for appropriately scaled-down parameters.

\subsection{Overview}
In Section~\ref{related-work}, we review related algorithms for network decomposition and coloring.

In Section~\ref{decomp-sec}, we state our network decomposition.

In Section~\ref{full-algorithm-sec}, we outline the full algorithm for list-coloring. It consists of two steps: an initial coloring step applied to all vertices, and multiple rounds of dense coloring.

In Section~\ref{sec:firstcoloring}, we describe the initial coloring step for creating the gap between the palette size and the degree for sparse vertices.

In Section~\ref{color-dense}, we describe a single round of the dense coloring procedure and analyze the behavior of the graph structure. 

In Section~\ref{sec:solver}, we finish our analysis by solving recurrence relations for dense components which yields the overall algorithm run time. 

In Section~\ref{list-color-locally-sparse-sec}, we apply the initial coloring step to give a full algorithm for locally-sparse graphs; this extends the algorithm of \cite{elk15} to list-coloring.

\section{Related work}\label{related-work}
A variety of network decompositions have been developed to solve distributed computing problems. Awerbuch et al.~\cite{awer89} introduced the notion of $(d,c)$-decompositions where each component has diameter~$d$ and the contracted graph is $c$-colorable. They give a deterministic procedure to obtain a $(2^{O(\sqrt{\log n \log \log n})}, 2^{O(\sqrt{\log n \log \log n})})$-decomposition, which can be used to deterministically compute a $(\Delta+1)$-coloring and MIS in $2^{O(\sqrt{\log n \log \log n})}$ rounds. Panconesi and Srinivasan~\cite{panc92} showed how to obtain a $(2^{O(\sqrt{\log n})}, 2^{O(\sqrt{\log n})})$-decomposition, yielding $2^{O(\sqrt{\log n})}$-time algorithms for $(\Delta+1)$-coloring and MIS. Linial and Saks~\cite{lin93} gave a randomized algorithm for obtaining a $(O(\log n), O(\log n))$-decomposition in $O(\log^2 n)$ rounds. Barenboim et al.~\cite{bar15} gave a randomized algorithm for obtaining $(O(1),O(n^\epsilon))$-decompositions in constant rounds.

Reed~\cite{ree98} introduced the structural decomposition to study the chromatic number of graphs of bounded clique size (see~\cite{MR01} for a detailed exposition). This was later used for applications including total coloring, frugal coloring, and computation of the chromatic number~\cite{ree99, mol14, mol98, mol10}. Our network decomposition method is inspired by theirs in the sense that they showed a graph can be decomposed into a sparse component and a number of dense components. However, as their main goal was to study the existential bounds, the properties of the decomposition between our needs are different. For example, the diameter is an important constraint in our case. Also, our decomposition must be computable in parallel, while theirs is obtained sequentially.

The $(\Delta + 1)$-coloring algorithms are briefly summarized in Table \ref{ResultsOverview}. Barenboim and Elkin's monograph~\cite{bar13} contains an extensive survey of coloring algorithms. Faster algorithms are available if we use more than $(\Delta+1)$ colors. For deterministic algorithms, Linial~\cite{linial92} and Szegedy and Vishwanathan~\cite{SV93} gave algorithms for obtaining a $O(\Delta^2)$-coloring in $O(\log^{*} n)$ rounds. Barenboim and Elkin~\cite{elk10} showed how to obtain an $O(\Delta^{1+\epsilon})$-coloring in $O(\log \Delta \cdot \log n)$ rounds. For randomized algorithms, Schneider and Wattenhofer~\cite{Sch10} showed that an $O(\Delta \log^{(k)}n + \log^{1+1/k} n )$-coloring can be obtained in $O(k)$ rounds. Combining the results in~\cite{Sch10} with Kothapalli et al.~\cite{KSOS06}, an $O(\Delta)$-coloring can be obtained in $O(\sqrt{\log n})$ rounds. Barenboim et al.~\cite{BEPS16} showed it can be improved to $2^{O(\sqrt{\log \log n})}$ rounds.

On the other hand, sparse-type graphs can be colored using significantly less than $(\Delta+1)$ colors. Panconesi and his co-authors~\cite{GP00, DGP98, GP97, PS97} developed a line of randomized algorithms for edge-coloring (i.e. coloring the line graph, which is sparse) and Brook-Vizing colorings in the distributed setting. For example, \cite{GP00} showed a $O( \frac{\Delta}{\log \Delta})$-coloring for girth-5 graphs in $O(\log n)$ rounds, provided $\Delta = (\log n)^{1 + \Omega(1)}$; this was generalized to triangle-free graphs by Pettie and Su~\cite{PS13}. The restriction on the size of $\Delta$ can be removed via the distributed Lov\'{a}sz Local Lemma~\cite{CPS17}.

Schneider et al.~\cite{sch13} investigated distributed coloring where the number of colors used depends on the chromatic number~$\chi(G)$. Their algorithm requires $(1-\frac{1}{O(\chi(G))})\cdot(\Delta+1)$ colors and running time of $O(\log \chi(G) + \log^* n)$ for graphs with $\Delta = \Omega(\log^{1+1/\log^* n} n)$ and $\chi(G) = O(\Delta/\log^{1+1/\log^* n} n)$.

More efficient algorithms for $(\Delta+1)$-coloring exist for very dense graphs, e.g.~a deterministic O($\log^* n$) algorithm for growth bounded graphs (e.g.~unit disk graphs)~\cite{sch10opt}, as well as for many types of sparse graphs~\cite{BEPS16,elk15,PS13}, e.g. for graphs of low arboricity. %The arboricity of a graph is the minimum number of edge-disjoint forests, whose union contains all edges of the graph.
Elkin et al.~\cite{elk15} described a $(\Delta+1)$-coloring algorithm for locally-sparse graphs. We will extend this result to cover list-colorings as well.

As we have discussed, the MIS problem and the coloring problems are related. An MIS can be computed in $O(\Delta+\log^{*} n)$ rounds deterministically ~\cite{BEK09} and in $2^{O(\sqrt{\log n})}$ rounds randomly~\cite{panc92}. More recently, Ghaffari~\cite{Moh16} reduced the randomized complexity of MIS to $O(\log \Delta) + 2^{O(\sqrt{\log \log n})}$. Whether an MIS can be obtained in polylogarithmic deterministic time or sublogarithmic randomized time remain interesting open problems.

\iffalse
A generalization of MIS, known as an \emph{ruling set}, has also been considered. A ($\alpha,\beta$)-ruling set~$U \subseteq V$ is a set of vertices such that two nodes $u,u' \in U$ have distance at least~$\alpha$ and for any node~$v \in V\setminus U$ there exists a node~$u\in U$ with distance at most~$\beta$~\cite{awer89}.  MIS is a special case, namely a $(2,1)$-ruling set.  A number of papers~\cite{gfeller07,sch13,awer89} use ruling sets to compute colorings in different kinds of graphs. A ruling set can be viewed as defining a network decomposition, such that any component has diameter at least~$\alpha$ and at most~$2\beta$.
\fi

\section{Network decomposition and sparsity} \label{decomp-sec}
In this section, we define a structural decomposition of the graph~$G$ into \emph{sparse} and \emph{dense} vertices. We measure  these notions with respect to a parameter~$\epsilon \in [0,1]$.

\begin{definition}[Friend edge]
An edge~$uv$ is a \emph{friend} edge if $u$ and $v$ share at least~$(1-\epsilon)\Delta$ neighbors, i.e. $|N(u) \cap N(v)| \geq (1-\epsilon) \Delta$. We define $F \subseteq E$ to be the set of friend edges.

For any vertex~$u$, vertex~$v$ is a friend of $u$ if~$uv \in F$; we denote the friends of $u$ by~$F(u)$.
\end{definition}

\begin{definition}[Dense and sparse vertices]
A vertex~$v \in V$ is \emph{dense} if it has at least~$(1-\epsilon)\Delta$ friends. Otherwise, it is \emph{sparse}.

We write $V^{\text{dense}} \subseteq V$ for the set of dense vertices in~$G$, and $V^{\text{sparse}}$ for the set of sparse vertices in~$G$.
\end{definition}

Next, we define the \emph{weak diameter}; this measures the diameter of a subgraph, while allowing shortcuts using nodes from the original graph.
\begin{definition}[Weak diameter]
Let $H \subseteq G$ be an induced subgraph of~$G$. For vertices~$u,v\in V$, let~$d(u,v)$ denote the distance between $u$ and $v$ in~$G$. The weak diameter of $H$ is defined to be $\max_{u,v \in H} d(u,v)$.
\end{definition}

Let $C_1, \ldots, C_k$ be the connected components of the subgraph $H = (V^{\text{dense}}, E_H ) \subseteq G$, where $E_H = \{uv \mid \mbox{$u,v \in V^{\text{dense}}$ and $uv \in F$} \}$. That is, they are the connected components induced by friend edges and dense vertices. The vertices of $G$ are partitioned disjointly as $V = V^{\text{sparse}} \sqcup V^{\text{dense}} = V^{\text{sparse}} \sqcup C_1 \sqcup \dots \sqcup C_k$. We refer to each component~$C_j$ as an \emph{almost-clique}. See Figure~\ref{fig:decomposition}.

\begin{figure}[t]
\centering
\includegraphics[scale = 0.35]{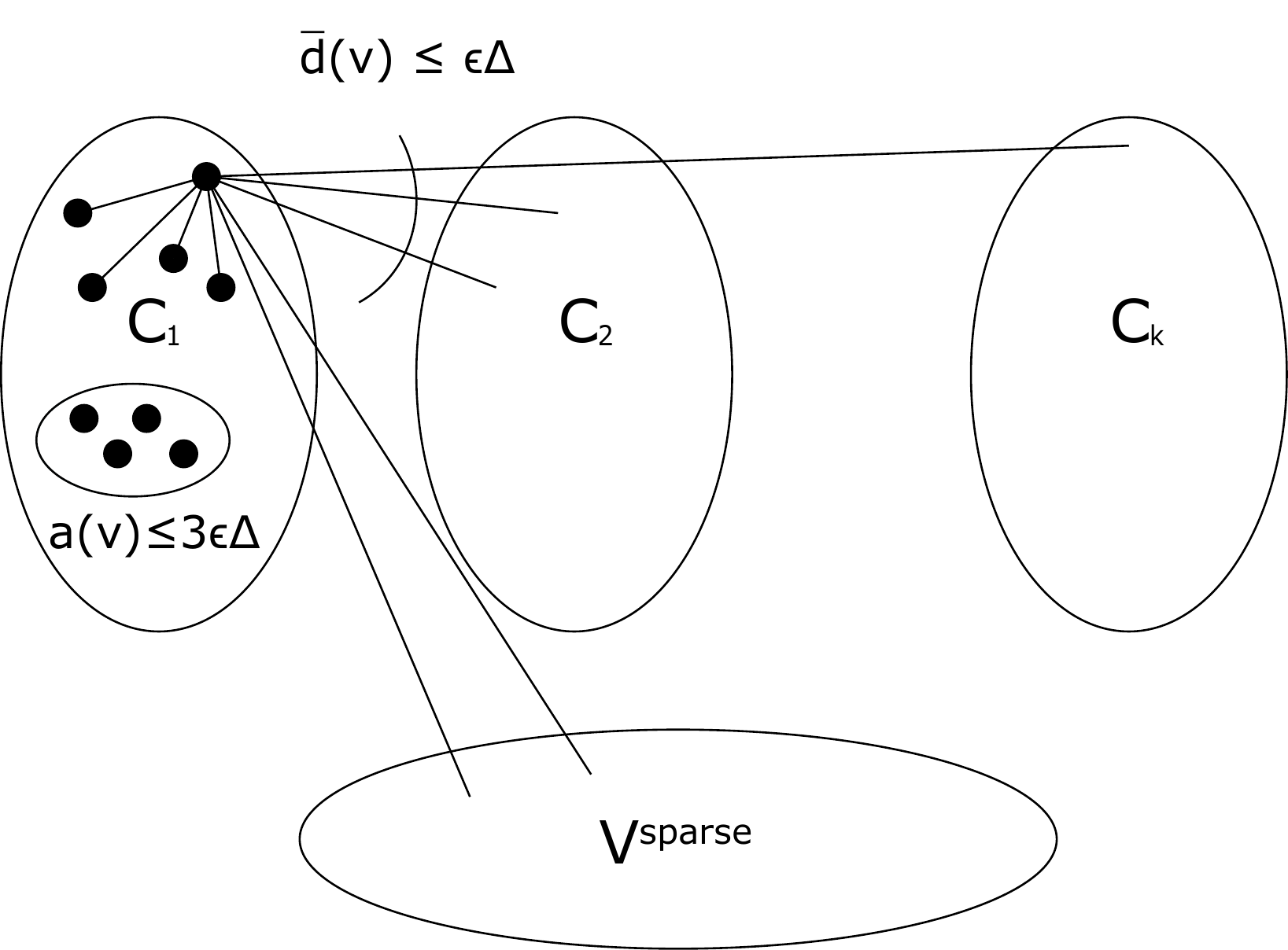}
\caption{An illustration of a network decomposition.}
\label{fig:decomposition}
\end{figure}

\begin{lemma}\label{lem:diameter}
Suppose $\epsilon < 1/5$. Then, for any vertices $x, y \in C_j$, we have $|N(x) \cap N(y)| \geq (1 - 2 \epsilon) \Delta$.
\end{lemma}

\begin{proof}
As $x, y$ are in the same component~$C_j$, there is a path of friend edges $x = u_0, \dots, u_t = y$ connecting them. We claim that $|N(x) \cap N(u_i)| \geq (1 - 2 \epsilon) \Delta$ for all~$i \geq 1$. We will show this by induction on~$i$. The base case~$i = 1$ follows as $x u_1$ is a friend edge.

Now, consider the induction step. As $u_{i-1} u_i$ is a friend, $|N(u_i) \cap N(u_{i-1})| \geq (1 - \epsilon) \Delta$. By the induction hypothesis, $|N(x) \cap N(u_{i-1})| \geq (1 - 2 \epsilon) \Delta$. We thus have:
\begin{align*}
| N(x) \cap N(u_i) | &\geq | N(x) \cap N(u_{i-1}) \cap N(u_i) | \\
&= | N(u_{i-1}) \cap N(u_i) | + |N(u_{i-1}) \cap N(x) | - | (N(u_{i-1}) \cap N(u_i)) \cup ( N(u_{i-1}) \cap N(x)) | \\
&\geq | N(u_{i-1}) \cap N(u_i) | + |N(u_{i-1}) \cap N(x) | - | N(u_{i-1}) | \\
&\geq (1 - \epsilon) \Delta + (1 - 2 \epsilon)\Delta - \Delta = (1 - 3 \epsilon) \Delta
\end{align*}

Since $x$ and $u_i$ are dense, we have $|N(x) \setminus F(x)| \leq \epsilon\Delta$ and $|N(u_i) \setminus F(u_i)| \leq \epsilon\Delta$. Therefore, $|F(x) \cap F(u_i)|  = |(N(x) \cap N(u_i))   \setminus (N(x) \setminus F(x)) \setminus (N(u_i) \setminus F(u_i)) |  \geq (1-3\epsilon)\Delta - \epsilon\Delta -\epsilon\Delta \geq (1-5\epsilon)\Delta > 0$.

So $x$ and $u_i$ have a common friend~$w$, such that $|N(x) \cap N(w)| \geq (1-\epsilon)\Delta$ and $|N(u_i) \cap N(w)| \geq (1-\epsilon)\Delta$. So:
\begin{align*}
| N(x) \cap N(u_i) | &\geq | N(x) \cap N(w) \cap N(u_i) | \\
&= | N(w) \cap N(u_i) | + |N(w) \cap N(x) | - | (N(w) \cap N(u_i)) \cup ( N(w) \cap N(x)) | \\
&\geq | N(w) \cap N(u_i) | + |N(w) \cap N(x) | - | N(w) | \\
&\geq (1 - \epsilon) \Delta + (1 - \epsilon)\Delta - \Delta = (1 - 2 \epsilon) \Delta
\end{align*}
\end{proof}

\begin{corollary}\label{cor:diameter}
Suppose $\epsilon < 1/5$. Then all almost-cliques have weak diameter of at most 2.
\end{corollary}
\begin{proof}
By Lemma~\ref{lem:diameter}, any vertices $x, y \in C_j$ have $|N(x) \cap N(y)| \geq (1 - 2\epsilon) \Delta > 0$. In particular, they have a common neighbor.
\end{proof}

A vertex~$v$ in $C_j$ can identify all other members of $C_j$ in $O(1)$ rounds by  the following procedure: Initially, each vertex~$u \in G$ broadcasts the edges incident to $u$ to all nodes within distance 3.  In this way, every vertex~$v$ learns the graph topology of all nodes up to distance 3, which is sufficient to determine whether an edge (both of whose endpoints are within distance 2 of $v$) is a friend edge and whether a vertex (within distance 2) is dense. Since by Corollary \ref{cor:diameter}, all members of $C_j$ are within distance 2 to~$v$, all the members can be identified. Also, the leader of $C_j$ can be elected as the member with the smallest ID.

\begin{definition}[External degree]
For any dense vertex $v \in C_j$, we define~$\extd(v)$, the \emph{external degree} of~$v$, to be the number of \emph{dense} neighbors of $v$ outside~$C_j$. (Sparse neighbors are not counted.)
\end{definition}

\begin{lemma}
\label{network1}
Any dense vertex $v$ has $\extd(v) \leq \epsilon \Delta$.
\end{lemma}
\begin{proof}
Let $v \in C_j$. As $v$ is dense, it has at least~$(1 - \epsilon) \Delta$ friends. So it has at most $\epsilon \Delta$ dense vertices which are not friends.  If any dense vertex~$w$ is a friend of~$v$, then by definition~$w \in C_j$. So $v$ has at most $\epsilon \Delta$ dense neighbors outside~$C_j$.
\end{proof}

\begin{definition}[Anti-degree]
For any dense vertex $v \in C_j$, we define the \emph{anti-degree} of $v$ to be $a(v) = |C_j \setminus N(v)|$.
\end{definition}

\begin{lemma}
\label{network2}
Suppose $\epsilon < 1/5$. Then any dense vertex~$v$ has $a(v) \leq 3 \epsilon \Delta$.
\end{lemma}
\begin{proof}
Let $v \in C_j$, and let~$R$ denote the number of length-$2$ paths of the form~$v, x, u$ where $x \in G$ and $u \in C_j \setminus N(v).$ We show this result by counting $R$ into ways different ways. First, observe that for any $u \in C_j \setminus N(v)$, there are precisely $|N(v) \cap N(u)|$ possibilities for the middle vertex~$x$. Lemma~\ref{lem:diameter} shows that:
$$
R = \sum_{u \in C_j \setminus N(v)} |N(v) \cap N(u)| \geq \sum_{u \in C_j \setminus N(v)} (1 - 2 \epsilon) \Delta = a(v) (1 - 2 \epsilon) \Delta
$$

We can also count~$R$ by summing over the middle vertex $x$:
\begin{align*}
R &= \sum_{x \in N(v)} |N(x) \cap ( C_j \setminus N(v))| \leq \sum_{x \in N(v)} |N(x) \setminus N(v)| \\
&= \sum_{x \in F(v)} |N(x) \setminus N(v)| +  \sum_{x \in N(v) \setminus F(v)} |N(x) \setminus N(v)| \\
&\leq \sum_{x \in F(v)} \epsilon \Delta +  \sum_{x \in N(v) \setminus F(v)} \Delta \\
&\leq \epsilon \Delta^2 + (1 - \epsilon)\Delta \qquad \text{(since $v$ is dense, $|F(v)| \geq (1 - \epsilon) \Delta$)} \\
&= \epsilon\Delta^2 (2 - \epsilon) 
\end{align*}
Combining these two inequalities gives $a(v) \leq \Delta \frac{\epsilon (2 - \epsilon)}{1 - 2 \epsilon}$; this is at most $3 \epsilon \Delta$ for $\epsilon < 1/5$.
\end{proof}

\begin{corollary}
For $\epsilon < 1/5$, all almost-cliques have size at most $(1 + 3 \epsilon) \Delta$.
\end{corollary}
\begin{proof}
Let $v \in C_j$. Then $|C_j| = |C_j \setminus N(v)| + |C_j \cap N(v)| \leq a(v) + |N(v)| \leq (1 + 3 \epsilon) \Delta$.
\end{proof}

\section{Full algorithm outline}
\label{full-algorithm-sec}
We can now describe our complete algorithm for list-coloring graphs, whether sparse or dense. It will be convenient to have a ``blank'' color, which is available to all vertices and which we denote~$0$; we say that $\chi(v) = 0$ to indicate that $v$ has not (yet) chosen a color. Our parameters will be specified in terms of a constant $K$; we require $K$ to be sufficiently large, and will state certain conditions on its value later in the proofs. We assume that $\epsilon^4 \Delta \geq K \ln n$. If $\epsilon^{4} \Delta <  K \ln n$, then $\Delta < \text{polylog}(n)$, and so that the coloring procedure of~\cite{BEPS16} will already color the graph in $O(\log \Delta) +2^{O(\sqrt{\log \log n})} = 2^{O(\sqrt{\log \log n})}$ rounds.

We set the density parameter to be 
$$
\epsilon =  \frac{100^{-\sqrt{\ln \Delta}}}{100 K}.
$$ 

Note that our assumptions ensure that $\epsilon < 1/5$, which is needed for the results of Section~\ref{decomp-sec} and elsewhere. The bound on $\epsilon$ is used at a number of other places without further comment.

\begin{algorithm}[H]
\begin{algorithmic}[1]
\STATE Decompose $G$ into $V^{sparse}, C_1, \ldots, C_k$.
\STATE Execute {\it the initial coloring step}:
\begin{ALC@g}
\FOR{all vertices $v$}
%\STATE With probability~$1/100$, select some color~$A(v)$ uniformly from~$\pal(v)$;
%\STATE Otherwise (with probability $99/100$) do nothing (we denote this by $A(v) = 0$).
\STATE $v$ selects a tentative color~$A(v)$ as follows:
$$
A(v) = 
\begin{cases}
0,              & \text{with probability } 99/100 \\
\text{color } c \in \pal(v) \text{ chosen uniformly at random},& \text{with probability } 1/100
\end{cases}
$$
\IF{no neighbor~$w \in N(v)$ has $A(w) = A(v)$ and $A(v) \neq 0$}
\STATE $v$ commits to permanent color $\chi(v) = A(v)$.
\ENDIF
\ENDFOR
\end{ALC@g}
\FOR{$i = 1, \dots, \lceil \sqrt{\ln \Delta} \rceil $}
\STATE Execute the \emph{dense coloring step} on the dense vertices (Algorithm \ref{algo:app1})
\ENDFOR
\STATE Run the algorithm of~\cite{elk15} to color the sparse vertices.
\STATE Run the algorithm of~\cite{BEPS16} to color the residual graph.
\end{algorithmic} 
\caption{The coloring algorithm}
\label{algo:app0}
\end{algorithm}

The key subroutine for the coloring algorithm is the dense coloring step. The $i^{\text{th}}$ dense coloring step is defined in terms of a parameter $\gamma_i \in [0,1]$, which we will specify shortly.
\begin{algorithm}[H]
\begin{algorithmic}[1]
\STATE Initialize $A(v) = 0$ for all $v \in V^{\text{dense}}$.
\STATE \textbf{For each}~$C_j$, elect a leader~$\ell_j$ to simulate the following process to color $C_j$:
\begin{ALC@g}
\STATE Generate a random permutation~$\pi_j$ of $C_j$
\FOR{$k= 1, \dots, L_j = \big \lceil |C_j| \gamma_i \big \rceil$}
\STATE Vertex $v = \pi_j(k)$ selects $A(v)$ uniformly at random from 
$$
\pal(v) \setminus \{ A(\pi_j(1)), \dots, A(\pi_j(k-1)) \}.
$$
\ENDFOR
\end{ALC@g}
\FOR{all $C_j$ and all $v \in C_j$}
\IF{no dense neighbor~$w \in N(v)$ has $A(w) = A(v)$, where $w \in C_{j'}$ and $\ID(\ell_{j'}) < \ID(\ell_j)$}
\STATE $v$ commits to permanent color $\chi(v) = A(v)$.
\ENDIF
\ENDFOR
\end{algorithmic}
\caption{The dense coloring step}
\label{algo:app1}
\end{algorithm}

We note that the decomposition of $G$ in Line 1 of Algorithm~\ref{algo:app0} remains fixed for the entire algorithm and all its dense coloring steps. Although in later steps vertices become colored and are removed from~$G$, we always define the decomposition in terms of the \emph{original} graph~$G$, not the residual graph. However, we abuse notation so that when we refer to a component~$C_j$ during an intermediate step, we mean the intersection of $C_j$ with the residual (uncolored) vertices. 

The algorithm is based on maintaining a partial coloring and a residual palette. That is, whenever a vertex is colored, we remove it from the graph as well as the vertex sets $V^{\text{sparse}}, C_1, \ldots ,C_k$; also, its selected color is removed from the palettes of its neighbors. An important parameter for such partial-coloring algorithms is the difference between the available colors, i.e. the palette size~$Q(v)$ and the uncolored neighbors~$d(v)$. We call this parameter \emph{(color) surplus} $S(v)$ of a vertex~$v$, defined by
$$
S(v) = Q(v) - d(v)
$$

The surplus $S(v)$ is initially at least one for every vertex, since all palettes are of size~$\Delta+1$. A vertex~$v$ can only lose a color from its palette if a neighbor becomes colored and drops out of the residual graph. Therefore, $S(v)$ can never decrease during the partial coloring.

For any vertex~$v$, let~$\pal_0(v), d_0(v), S_0(v)$ denote, respectively, the palette,  degree, and surplus of $v$ after the initial coloring step (note that the 0 in the subscript does not denote time 0, but the time immediately after the initial coloring step). We set $Q_0(v) = |\pal_0(v)|$. We will show in Theorem~\ref{initial-coloring-good-prop} that w.h.p. every sparse vertex~$v$ has $S_0(v) = \Omega(\epsilon^2 \Delta)$ and that every vertex~$v$ (sparse or dense) has~$Q_0(v) \geq \Delta /2$.

Then we turn our attention to the dense vertices and we will show that they can be colored efficiently. For a dense vertex $x \in C_j$, we let~$\extd_0(x)$ and $a_0(x)$ denote its external degree and anti-degree after the initial coloring step. Let $\pal_i(x), d_i(x), \extd_i(x), a_i(x)$, and $Q_i(x)$ denote the quantities at the end of the $i^{\text{th}}$ dense coloring step. As we color the graph, we maintain two key parameters, $D_i$ and~$Z_i$, which bound the external degree, anti-degree, and palette size for dense vertices after the $i^{\text{th}}$ dense coloring step. Specifically, we ensure the invariant that every dense vertex $v$ has
$$
\extd_i(v) \leq D_i, a_i(v) \leq D_i, Q_i(v) \geq Z_i
$$

We will then set our parameter $\gamma_i$ by
$$
\gamma_i = 1 - 2 \sqrt{D_{i-1}/Z_{i-1}};
$$
We will show in Corollary~\ref{gamma-cor} that $\gamma_i \in [0,1]$ as required.
 
At the end of the dense coloring steps, every sparse vertex~$v$ has
$$
S_{\lceil \sqrt{\ln \Delta} \rceil} (v) \geq S_0(v) = \Omega(\epsilon^2 \Delta)
$$

The algorithm of Elkin et al.~\cite{elk15} is designed for list-coloring where vertices have a large surplus, which indeed holds for the sparse vertices. Thus they can be colored in $O( \log(1/\epsilon^2)) + 2^{O(\sqrt{\log \log n})} = O( {\sqrt{\log \Delta}}) + 2^{O(\sqrt{\log \log n})}$ rounds. This removes the sparse vertices from the graph, leaving only the dense vertices behind.

After the sparse vertices are removed, Theorem~\ref{th1} shows that each remaining dense vertex is connected to $\Delta' = O(\log n) \cdot 2^{O(\sqrt{\log \Delta})}$ other vertices. The algorithm of~\cite{BEPS16} then takes $O(\log \Delta') + 2^{O(\sqrt{\log \log n})} = O(\sqrt{\log \Delta}) + 2^{O(\sqrt{\log \log n})}$ steps.

\section{The initial coloring step}
\label{sec:firstcoloring}
The initial coloring step is designed to achieve two important objectives:
\begin{enumerate}
\item[(A1)] For every vertex~$v$, we have $Q_0(v) \geq \Delta/2$.
\item[(A2)] For every sparse vertex~$v$, we have $S_0(v) =  \Omega(\epsilon^2 \Delta)$.
\end{enumerate}
Recall that $Q_0$ and $d_0$ are respectively the palette size and degree of vertex~$v$ after the initial coloring step, and the (color) surplus is $S_0(v) = Q_0(v) - d_0(v)$. Property (A1) is fairly straightforward, and most of our effort will be to show that (A2) holds w.h.p. for a given sparse vertex~$v$.

We briefly summarize the initial coloring step. With probability~$\alpha = \tfrac{1}{100}$, each vertex~$v$ chooses a tentative color~$A(v)$ uniformly at random from its palette.\footnote{We will write $\alpha$ instead of directly writing $1/100$.} It discards the tentative color if any neighbor also chooses the same tentative color; in this case we say that $v$ is \emph{de-colored}. We let~$\chi(v)$ denote the permanent color selected by $v$ \emph{after the initial coloring step.} We say that $A(v)= 0$ if vertex~$v$ chose not to select a color initially and we say~$\chi(v)=0$ if $v$ is uncolored (either because it did not select a tentative color, or it became de-colored).

For each color~$c$ and each vertex $v$, we define $N_c(v)$ to be the set of neighbors of $v$ whose palette contains~$c$; that is, $$
N_c(v) = \{w \in N(v) \mid c \in \pal(w) \}
$$

Let us now fix some sparse vertex~$v$, and show that after the initial coloring step~$v$ has a large surplus. We define a color~$c$ to be \emph{good} if
$$
\sum_{w \in N(v)} [\chi(w) = c] \geq 1 + [c \in \pal(v)]
$$
(Here and in the remainder of the paper, we use the Iverson notation so that for any predicate $\mathcal P$,  $[ \mathcal P ]$ is equal to 1 if $\mathcal P$ is true and zero otherwise.)

Let~$J$ denote the set of colors that are good for~$v$. 
\begin{proposition}
\label{sparse-prop2}
The following bound holds with probability one:
$$
S_0(v) \geq |J|
$$
\end{proposition}
\begin{proof}
Suppose $c \notin \pal(v)$ and there is some neighbor~$w \in N(v)$ with~$\chi(w) = c$. When we remove~$w$ from the residual graph, $\pal(v)$ does not change while~$d(v)$ decreases by one. Suppose $c \in \pal(v)$, and there are two neighbors $w_1, w_2 \in N(v)$ with $\chi(w_1) = \chi(w_2) = c$. When we remove~$w_1$ and~$w_2$ from the residual graph, $\pal(v)$ decreases by one while~$d(v)$ decreases by two.
\end{proof}

In light of Proposition~\ref{sparse-prop2}, it will suffice to show that $|J|$ is large with high probability. We do so in two stages: first, we show that $\bE[ |J| ]$ is large, and second we show that $|J|$ is concentrated around its mean.
\begin{lemma} \label{sparse-prop3a}
For $c \notin \pal(v)$, we have $P(c \in J) = \Omega \bigl( |N_c(v)|/\Delta \bigr)$.
\end{lemma}
\begin{proof}
For each $x \in N_c(v)$, let us define $B_x$ to be the event that (i) $A(x) = c$ and (ii) $A(z) \neq c$ for all $z \in N(v) \cup N(x) \setminus \{x \}$. 

If event $B_x$ occurs, then we will have $\chi(x) = c$, and so $c$ will go into $J$. Also, condition (ii) ensures that the events $B_x$ are all mutually exclusive. Thus
\begin{align*}
P(c \in J) &\geq P( \bigvee_{x \in N_c(v)} B_x ) = \sum_{x \in N_c(v)} P(B_x)  \\
&= \sum_{x \in N_c(v)} \frac{\alpha}{\Delta+1} \Bigl( 1 - \frac{\alpha}{\Delta+1} \Bigr)^{|N_c(x) \cup N_c(v) \setminus \{ x \} |} \geq \sum_{x \in N_c(v)} \frac{\alpha}{\Delta+1} \Bigl( 1 - \frac{\alpha}{\Delta+1} \Bigr)^{2 \Delta - 1} \\
&\geq \frac{|N_c(v)| \alpha (1 - \alpha)^2}{\Delta} = \Omega\bigl( |N_c(v)| / \Delta \bigr) 
\end{align*}
\end{proof}

\begin{lemma}
\label{sparse-prop3}
Suppose that $\epsilon \Delta \geq 3$. If $c \in \pal(v)$ and $|N_c(v)| \geq (1 - \epsilon/2) \Delta$, then $P( c \in J ) = \Omega(\epsilon^2)$.
\end{lemma}
\begin{proof}
Define $U$ to be the set of ordered pairs $(x,y)$ with $x, y \in N_c(v), x < y$ and $xy \notin E$. For every pair of vertices $(x,y) \in U$, let us define $B_{x,y}$ to be the event that (i) $A(x) = A(y) = c$ and (ii) $A(z) \neq c$ for all $z \in N(v) \cup N(x) \cup N(y) \setminus \{x, y\}$. 

If event~$B_{x,y}$ occurs, then (as $x,y$ are not neighbors) we will have $\chi(x) = \chi(y) = c$, and so $c$ will go into~$J$. Also, note that condition (ii) ensures that the events $B_{x,y}$ are mutually exclusive. Thus,
\begin{align*}
P( c \in J) &\geq P( \bigvee_{(x,y) \in U} B_{x,y}) = \sum_{(x,y) \in U} P(B_{x,y}) \\
&= \sum_{(x,y) \in U} \frac{\alpha^2}{(\Delta+1)^2} \bigl( 1 - \frac{\alpha}{\Delta+1} \bigr)^{|N_c(v) \cup N_c(x) \cup N_c(y) \setminus \{x,y \} |} \\
&\geq \sum_{(x,y) \in U} \frac{\alpha^2}{(\Delta+1)^2} \bigl( 1 - \frac{\alpha}{\Delta+1} \bigr)^{3 \Delta - 2} \geq \sum_{(x,y) \in U} \frac{\alpha^2}{(\Delta+1)^2} (1 - \alpha)^3 = \Omega(\frac{|U|}{\Delta^2})
\end{align*}

Next, we estimate~$|U|$. Consider the set $W = \{w \in N_c(v) \mid wv \notin F \}$. By definition of sparsity, $v$ has at most $(1 - \epsilon) \Delta$ friends. Thus, $|W| \geq |N_c(v)| - (1 - \epsilon) \Delta \geq \epsilon \Delta/2$.

Since each $w \in W$ is not a friend of~$v$, we have $|N(w) \cap N(v)| < (1 - \epsilon) \Delta$.  So $|N_c(v) \setminus ( N(w) \cup \{w \})| \geq |N_c(v)| - 1 - |N(v) \cap N(w)| > \epsilon \Delta/2 - 1$. So we count~$U$ as follows:
\begin{align*}
|U| &= \sum_{\substack{x \in N_c(v) \\ y \in N_c(v) \setminus N(x) \\ y < x}} \negthickspace \negthickspace \negthickspace 1 
= \sum_{\substack{x \in N_c(v) \\ y \in N_c(v) \setminus (N(x) \cup \{ x \})}} \negthickspace \negthickspace \negthickspace \negthickspace \negthickspace \tfrac{1}{2} 
\geq \sum_{\substack{x \in W \\ y \in N_c(v) \setminus (N(x) \cup \{x \})}} \negthickspace \negthickspace  \negthickspace \negthickspace \tfrac{1}{2} \\
&\geq \tfrac{1}{2} \sum_{x \in W} ( \epsilon \Delta/2 - 1 ) \geq \tfrac{1}{2} \cdot (\epsilon/2) \Delta \cdot (\epsilon \Delta/2 - 1) \\
&= \Omega(\epsilon^2 \Delta^2) \qquad \text{(as $\epsilon \Delta \geq 3$ and $\epsilon < 1/5$)}
\end{align*}

This shows that $P(c \in J) = \Omega(\frac{\epsilon^2 \Delta^2}{\Delta^2}) = \Omega(\epsilon^2)$.
\end{proof}

\begin{lemma}
\label{sparse-prop4}
Suppose that $\epsilon \Delta \geq 3$.  If $d(v) \geq (1 - \epsilon/4) \Delta$, then $\bE [ |J| ] = \Omega(\epsilon^2 \Delta)$.
\end{lemma}
\begin{proof}
Let us partition the set of colors~$\mathcal C$ into three disjoint sets:
\begin{align*}
B_1 &= \{ c \in \mathcal C \mid c \notin \pal(v) \} \\
B_2 &= \{ c \in \mathcal C \mid c \in \pal(v), |N_c(v)| \geq (1 - \epsilon/2) \Delta \} \\
B_3 &= \{ c \in \mathcal C \mid c \in \pal(v), |N_c(v)| < (1 - \epsilon/2) \Delta \}
\end{align*}

If $|B_2| \geq \Delta/4$, then by Proposition~\ref{sparse-prop3} we immediately have 
$$
\bE[ |J| ] \geq \sum_{c \in B_2} P(c \in J) \geq \Delta/4 \cdot \Omega(\epsilon^2)
$$
and we are done. So let us suppose that $|B_2| < \Delta/4$. 

Recall that $|\pal(w)| = \Delta+1$ for every vertex $w$. Thus, for each $w \in N(v)$, there are exactly $\Delta+1$ values of $c$ with~$w \in N_c(v)$. By double counting, 
\begin{align*}
(\Delta+1) d(v) &= \sum_{c \in \mathcal C} |N_c(v)| = \sum_{c \in B_1} |N_c(v)| + \sum_{c \in B_2} |N_c(v)| + \sum_{c \in B_3} |N_c(v)| \\
&\leq \sum_{c \in B_1 } |N_c(v)| + |B_2| \Delta + |B_3| (1 - \epsilon/2) \Delta
\end{align*}

Rearranging, and using the fact that $|\pal(v)| = |B_2| + |B_3| = \Delta+1$, gives:
\begin{align*}
\sum_{c \in B_1} |N_c(v)| &\geq (\Delta + 1) d(v) - |B_2|\Delta - |B_3|(1-\epsilon/2)\Delta \\
&\geq (\Delta + 1)(1-\epsilon/4)\Delta - |B_2|\Delta - (\Delta + 1 - |B_2|)(1-\epsilon/2)\Delta \\
&= \tfrac{\epsilon \Delta}{4} ( \Delta + 1 - 2 |B_2|) \geq \epsilon \Delta^2/16
\end{align*}
where the last inequality holds since $|B_2| < \Delta/4$.

Lemma~\ref{sparse-prop3a} then gives:
\begin{align*}
\bE[ |J| ] \geq \sum_{c \in \mathcal C \setminus \pal(v)} P(c \in J) \geq \sum_{c \in B_1} \Omega( |N_c(v)|/\Delta ) = \Omega( \epsilon^2 \Delta ).
\end{align*}
\end{proof}

\begin{lemma}
\label{sparse-conc-prop}
With probability at least $1 - e^{-\Omega(\epsilon^4 \Delta)}$, we have $S_0(v) = \Omega(\epsilon^2 \Delta)$.
\end{lemma}
\begin{proof}
If $d(v) \leq 1 - \epsilon \Delta/4$ then~$S_0(v) \geq \epsilon \Delta/4 = \Omega(\epsilon^2 \Delta)$ with certainty. Also, if $\epsilon \Delta < 3$, then $S_0(v) \geq 1 = \Omega(\epsilon^2 \Delta)$ with certainty. So let us assume that $d(v) > 1 - \epsilon \Delta/4$ and $\epsilon \Delta > 3$. We will now show~$|J| = \Omega(\epsilon^2 \Delta)$ with probability at least~$1-e^{-\Omega(\epsilon^4\Delta)}$. By Proposition~\ref{sparse-prop2} this will establish the result.

Let $W = \{v\} \cup N(v)$ and let~$U$ denote the set of vertices with distance 2 to~$v$. Let us define $J'$ to be the set of colors which would be good \emph{if~$A(x) = 0$ for all~$x \in U$}. Note that vertices in $U$ can de-color vertices in $W$. So vertices of $U$ can only remove colors from $J$ and hence  $J \subseteq J'$.

Since $J \subseteq J'$, Proposition~\ref{sparse-prop4} shows that 
$$
\bE[|J'|] \geq \phi \epsilon^2 \Delta
$$
for some constant~$\phi > 0$. 

Note that for~$w \in W$, modifying the value of~$A(w)$ can only change~$|J'|$ by at most 2 (the value of~$A(w)$ only affects whether~$A(w) \in J'$). Hence, by the Bounded Differences Inequality,
$$
P \Bigl(  |J'| < \frac{\phi \epsilon^2 \Delta}{2} \Bigr) \leq \exp\left(-\frac{ (\phi \epsilon^2 \Delta)^2}{ 2 \cdot \sum_{i \in \{v \} \cup N(v)} 2^2 }\right) \leq \exp(-\Omega(\epsilon^4 \Delta))
$$

Now, let us condition on the full set of values $A(w)$ for $w \in W$, and also condition on the event $|J'| \geq \tfrac{\phi}{2} \epsilon^2 \Delta$ (which only depends on the values of $A(w)$ for $w \in W$). Let~$j = |J'|$. The values of $A(u)$ for $u \in U$ are still independent and uniform, and each such vertex has the possibility of de-coloring a vertex in~$W$.

For each $c \in J' \setminus  \pal(v)$, let~$y_c$ be any vertex in $N(v)$ with $A(y_c) = c$ and not de-colored by any vertices in~$W$. Similarly, for $c \in J' \cap \pal(v)$, let~$y_c, y'_c$ be any two vertices in $N(v)$ with $A(y_c) = A(y'_c) = c$ and not de-colored by any vertices in $W$ (so $y_c$ and $y'_c$ cannot be neighbors). Such colors will go into $J$ unless a vertex in $N(y_c)$ selects $c$ (respectively, in $N(y_c) \cup N(y'_c)$ selects color~$c$.)

If a vertex~$u \in U$ selects $A(u) = c$ for such a color~$c \in J'$, causing color~$c$ to not appear in~$J$, we say that $u$ \emph{disqualifies} color~$c$. Define the random variable~$R$ by
$$
R = \sum_{\substack{c \in J' \\ u \in (N(y_c) \cup N(y'_c))\cap U}} \negthickspace \negthickspace \negthickspace \negthickspace \negthickspace \negthickspace [ \text{$u$ disqualifies color $c$}]
$$

Observe that $|J| \geq |J'| - R = j - R$, and so the Lemma will follow by showing that $R < j/2$.

Each vertex~$u \in U$ disqualifies any given color~$c$ with probability at most~$\frac{\alpha}{\Delta+1}$. Furthermore there are at most $2 \Delta$ vertices $u \in U$ that can disqualify any given color $c \in J'$. Hence,
\begin{align*}
\bE[R] &\leq \sum_{\substack{c \in J' \\ u \in (N(y_c) \cup N(y'_c))\cap U}} \negthickspace \negthickspace \negthickspace \negthickspace \negthickspace \negthickspace P( \text{$u$ disqualifies color $c$}) \leq j \cdot 2 \Delta \cdot \frac{\alpha}{\Delta+1} \\
& \leq j/4 \qquad \text{as $\alpha = 1/100$}
\end{align*}

All such disqualification events are negatively correlated. Using the fact that $j \geq \frac{\phi \epsilon^2 \Delta}{2}$, we apply Chernoff's bound to obtain
$$
P( R \geq j/2)\leq e^{-(j/4) \cdot 4 / 3} \leq \exp(-\frac{\phi \epsilon^2 \Delta}{6}) \leq \exp(-\Omega(\epsilon^2 \Delta))
$$

Overall, we have shown that $|J| = \Omega(\epsilon^2 \Delta)$ with probability $1 - \exp(-\Omega(\epsilon^4 \Delta))$.
\end{proof}

\begin{theorem}
\label{initial-coloring-good-prop}
For $K$ a sufficiently large constant, properties (A1) and (A2) hold for every vertex w.h.p.
\end{theorem}
\begin{proof}
By Lemma~\ref{sparse-conc-prop}, for any individual sparse vertex~$v$ the probability that (1) fails is at most~$e^{-\Omega(\epsilon^4 \Delta)}$. Since $\epsilon^4 \Delta \geq K \ln n$, this is at most~$n^{-K'}$ for an arbitrary large constant $K'>0$ given that $K$ is sufficiently large. (A2) follows by taking a union bound over all sparse vertices.

To show property (A1), fix some vertex~$v$, and note that any $w \in N(v)$ chooses an initial color with probability at most~$\alpha$, independently of any other vertices. Thus, a Chernoff bound shows that there is a probability of $e^{-\Omega(\Delta)}$ that more than $\Delta/2$ neighbors of $v$ are colored. So with probability $1 - e^{-\Omega(\Delta)}$, vertex~$v$ loses at most~$\Delta/2$ colors from its original palette of size of~$\Delta+1$. Again, this is at most~$n^{-K'}$ for an arbitrary constant~$K'>0$ given that~$K$ is sufficiently large.
\end{proof}

\section{Coloring the dense vertices}
\label{color-dense}
Suppose that we are at the beginning of the $i^{\text{th}}$ dense coloring step. We assume that there are parameters~$D_{i-1}$,~$Z_{i-1}$, such that all dense vertices $v$ have the following properties:
\begin{enumerate}
\item $a_{i-1}(v) \leq D_{i-1}$
\item ${\bar d}_{i-1}(v) \leq D_{i-1}$
\item $Q_{i-1}(v) \geq Z_{i-1}$
\end{enumerate}

Henceforth we will suppress the dependence on $i$ and write $D$, $Z$, $a(v)$, ${\bar d}(v)$, $\pal(v)$, and~$Q(v)$. We define $\delta = D/Z$ and $\gamma = (1 - 2 \sqrt{\delta})$.

Let us consider some almost-clique~$C_j$, with $M_j = |C_j|$ vertices. The dense coloring step for each $C_j$ generates a permutation~$\pi_j$ of its members. Starting from vertex~$\pi_j(1)$ up to vertex~$\pi_j(L_j)$, where $L_j = \lceil M_j \gamma \rceil$, each vertex selects a tentative color from its palette excluding the colors selected by lower rank vertices. (Note that the leader~$\ell_j$ in $C_j$ simulates this process.) Then, a vertex becomes \emph{de-colored} if an external neighbor from a lower indexed component chooses the same color.

Our goal is to show for some parameters $D'$ and $Z'$ that at the end of the round holds~$a_i(v) \leq D'$,~$\extd_i(v) \leq D'$ and~$Q_i(v) \geq Z'$. To do this, we will show that most vertices are colored in round~$i$.

We require throughout this section the following conditions on $D$ and~$Z$, which we will refer to as the \emph{regularity conditions}:
\begin{enumerate}
\item[(R1)] $D \delta \geq K \ln n$ for some sufficiently large constant~$K$
\item[(R2)] $\delta \leq 1/K$ for some sufficiently large constant~$K$
\end{enumerate}
Recall that $K$ is a universal constant that we will not explicitly compute. At several places we will assume it is sufficiently large. In Section~\ref{sec:solver} we will discuss how to satisfy these regularity conditions (or how our algorithm can succeed when they become false). \emph{In Section~\ref{color-dense}, we require implicitly that these regularity conditions all hold.} 

\subsection{Overview}
We first contrast our dense coloring procedure with a naive one, which assigns each vertex a random color and de-colors a vertex if there is a conflict. It is not hard to show that such a procedure successfully assigns a color to a vertex with constant probability. Thus, in each round, the degrees are shrinking by a constant factor in expectation. So it takes $\Omega(\log \Delta)$ rounds to reduce to a low (near-constant) degree.

In order to get a faster algorithm, we need to color much more than a constant fraction of all vertices per round. Here, the network decomposition plays the decisive role as only \emph{external} neighbors of a vertex~$v$ can de-color ~$v$. To illustrate, suppose that each vertex~$v$ selects a color from its palette uniformly at random. (That is, suppose we ignore the interaction between $v$ and the other vertices in $C_j$). Since the external neighbors are upper bounded by $D$ and the palette size is at least~$Z$, even if the external neighbors of $v$ choose distinct colors, the probability that $v$ has any conflicts with its neighbors is upper bounded by $D/Z = \delta = O(\epsilon)$. Ideally, we would like to show that each cluster shrinks by a factor of $\delta$ in each round. Moreover, one would also need to prove that the ratio~$D'/Z'$ in the next round remains approximately~$\delta$, so that the almost-cliques continue to shrink by the same factor.

The reason why we only attempt to color the first~$L_j$ vertices rather than the entire almost-clique is that we cannot afford the palette size to shrink too fast. A ``controlled'' uniform shrinking process maintains the overall ratio between palette size, external neighbors, and internal neighbors. This renders undesirable  scenarios very unlikely.

The following lemma uses the regularity conditions to show a useful bound on several parameters of the almost-clique.
\begin{lemma}
\label{aprop1}
For any $v \in C_j$, the following bound holds with probability one:
$$
Q(v) - L_j \geq Z \sqrt{\delta} + D
$$
\end{lemma}
\begin{proof}
Note that $M _j = |N(v) \cap C_j| + |C_j \setminus N(v)| \leq d(v) + a(v) \leq Q(v) + D$. Therefore,
{\allowdisplaybreaks
\begin{align*}
Q(v) - L_j &\geq Q(v) - \lceil M \gamma \rceil \geq Q(v) -  M_j (1 - 2 \sqrt{\delta})- 1 \\
&\geq Q(v) -  (Q(v)+D)(1 - 2 \sqrt{\delta})  - 1 \\
&= 2 Q(v) \sqrt{\delta} + 2 D \sqrt{\delta} - D - 1 \\
&\geq 2 Z \sqrt{\delta} + 2 D \sqrt{\delta} - D - 1 \qquad \text{as $Q(v) \geq Z$}  \\
&=Z \sqrt{\delta}  + D/\sqrt{\delta} + 2 D \sqrt{\delta} - D - 1 \qquad \text{ as $D = \delta Z$ by definition} \\
&\geq Z \sqrt{\delta}+ D \qquad \text{as $(1/\delta) \geq K$ and $D \geq 2$ by regularity conditions}
\end{align*}
}
\end{proof}

\subsection{Concentration of the number of uncolored vertices}
We will show that most vertices become colored at the end of a dense coloring step. We distinguish two ways in which a vertex~$v$ could fail to be colored: first, it may be \emph{de-colored} in the sense that it initially chose a color, but then had a conflict with an almost-clique of smaller index. Second, it may be \emph{initially-uncolored} in the sense that~$\pi_j^{-1} (v) > L_j$.

\begin{lemma}
\label{color-lemma}
Let $T = \{v_1, \dots, v_t \} \subseteq C_j$. Let $c_1, \dots, c_t$ be an arbitrary sequence of non-blank colors. Then
$$
P( A(v_1) = c_1 \wedge \dots \wedge A(v_t) = c_t  ) \leq (Z \sqrt{\delta})^{-t}
$$
\end{lemma}
\begin{proof}
Let us condition on the permutation~$\pi_j$, and without of loss generality $\pi^{-1}_j(v_1) < \pi^{-1}_j(v_2) < \dots < \pi^{-1}_j (v_t)$. We assume $\pi^{-1}_j(v_t) \leq L_j$ as otherwise $A(v_t) = 0$.

For each $i = 1, \dots, t$ we claim that $P(A(v_i) = c_i) \leq \frac{1}{Q(v_i) - L_j}$, even after conditioning on all the colors choices made by vertices $w$ with $\pi^{-1}_j(w) < \pi_j^{-1}(v_i)$. For, at this point, at most $\pi^{-1}_j(v_i) \leq L_j$ colors from the palette of $v_i$ have been used by previously-colored neighbors from $C_j$. Hence, $v_i$ has a remaining palette of size $Q(v_i) - L_j$. Using Lemma~\ref{aprop1} then gives
\begin{align*}
P(A(v_i) = c_i \mid A(v_1) = c_1 \wedge \dots \wedge A(v_{i-1}) = c_{i-1}) &\leq \frac{1}{Q(v_i) - L_j}  \leq \frac{1}{Z \sqrt{\delta}}
\end{align*}
\end{proof}

\begin{lemma}
\label{prop1}
Let $T \subseteq V^{\text{dense}}$. The probability that all the vertices in $T$ become de-colored is at most $(\sqrt{\delta})^{|T|}$.
\end{lemma}
\begin{proof}
Let us sort the almost-cliques $C_1, \dots, C_k$ by the vertex ID of their leaders, so that $\ell_1\leq\ell_2\leq~\dots~\leq~\ell_k$. For each $j= 1, \dots, k$ we define $T_j = T \cap C_j$. 

We will show for any $j$ that the vertices in $T_j$ become de-colored with probability at most $(\sqrt{\delta})^{|T_j|}$, conditioned on the event that the vertices in $T_1, \dots, T_{j-1}$ become de-colored.  In fact, we will not just condition on the event that the vertices in $T_1, \dots, T_{j-1}$ become de-colored, but we will condition on the complete set of random variables involved in $C_1, \dots, C_{j-1}$. (The event that $T_j$ becomes de-colored is a function of only the colors involved in $C_1, \dots, C_j$.)

For each $v \in T_j$, the event that $v$ becomes de-colored is a union of at most ${\bar d}(v) \leq D$ events of the form $\chi(v) = c$, where $c$ enumerates the colors selected by vertices in $N(v) \cap (C_1 \cup C_2 \cup \dots \cup C_{j-1})$. Hence, the event that all of the vertices in $T_j$ become de-colored is a union of $D^{|T_j|}$ events of the form stated in Lemma~\ref{color-lemma}, each of which has probability at most $(Z \sqrt{\delta})^{-|T_j|}$. Therefore, the probability that all of them become de-colored is $(\frac{D}{Z \sqrt{\delta}})^{|T_j|} = (\sqrt{\delta})^{|T_j|}$.
\end{proof}

\begin{lemma}
\label{prop2}
Let $T  \subseteq V^{\text{dense}}$. The probability that all of the vertices in $T$ are initially uncolored is at most $(2 \sqrt{\delta})^{|T|}$.
\end{lemma}
\begin{proof}
It suffices to show that for a particular~$C_j$, the probability that all vertices in $T_j = T \cap C_j$ are initially uncolored is bounded by $(2 \sqrt{\delta})^{|T_j|}$, since the nodes in distinct almost-cliques make their choices independently.

We select from~$C_j$ a set of $L_j$ vertices to be colored, uniformly at random without replacement. Thus, the probability that all vertices in $T_j$ are initially-uncolored is:
{\allowdisplaybreaks
\begin{align*}
P( \text{$T_j$ is initially uncolored}) &= \frac{\binom{M_j-|T_j|}{L_j}}{\binom{M_j}{L_j}} 
\end{align*}
If $|T_j| > M_j-L_j$, the right hand side is zero and we are done. Otherwise,
\begin{align*}
P( \text{$T_j$ is initially uncolored}) &= \Bigl(\frac{M_j - |T_j|}{M_j}\Bigr) \Bigl(\frac{M_j - |T_j| - 1}{M_j - 1}\Bigr) \dots \Bigl(\frac{M_j-L_j+1 - |T_j|}{M_j - L_j + 1}\Bigr) \\
&\leq \Bigl(\frac{M_j - 1}{M_j}\Bigr)^{|T_j|} \Bigl( \frac{M_j-2}{M_j-1} \Bigr)^{|T_j|} \dots \Bigl(\frac{M_j-L_j+1 - 1}{M_j - L_j + 1}\Bigr)^{|T_j|} \\
&=\Bigl(1 - \frac{L_j}{M_j} \Bigr)^{|T_j|} \leq \Bigl (1 - \frac{ M_j (1 - 2 \sqrt{\delta} ) }{M_j} \Bigr)^{|T_j|} = ( 2 \sqrt{\delta} )^{|T_j|}
\end{align*}
}
\end{proof}

\begin{lemma}
\label{prop3}
Let $T \subseteq V^{\text{dense}}$, and let~$s$ be a real number with $s \geq |T|, s \geq \frac{K \ln n}{\sqrt{\delta}}$. Then the probability that~$T$ contains more than $12 s \sqrt{\delta}$ uncolored vertices at the end of round~$i$ is at most~$n^{-K'}$ for an arbitrarily large constant $K'>0$.
\end{lemma}
\begin{proof}
Let $x = 6 s \sqrt{\delta}$. We claim that the number of de-colored vertices in $T$ is at most $x$ with probability~$1 - n^{-K'}/2$ and we also claim that the number of initially-uncolored vertices is at most $x$ with probability $1 - n^{-K'}/2$; combining these two claims gives the stated result.

The proofs are nearly the same, so we show only the latter one. By a union bound over all possible sets of size~$\lceil x \rceil$, the probability that the number of initially-uncolored vertices exceeds $\lceil x \rceil$ is at most
{\allowdisplaybreaks
\begin{align*}
\binom{|T|}{\lceil x \rceil} (2 \sqrt{\delta})^{\lceil x \rceil} &\leq \Bigl( \frac{e|T|}{\lceil x \rceil} \Bigr)^{\lceil x \rceil} \cdot (2 \sqrt{\delta})^{\lceil x \rceil} \\
&\leq \Bigl( \frac{e s (2 \sqrt{\delta})}{6 s \sqrt{\delta}} \Bigr)^{x} \qquad \text{as $x = 6 s\sqrt{\delta}$ and $|T| \leq s$}\\
&\leq \Bigl( \frac{2e}{6} \Bigr)^{K \ln n} \qquad \text{as $x \geq 6 s \sqrt{\delta} \geq K \ln n$} \\
&\leq n^{-K'}/2 \qquad \text{for $K$ a sufficiently large constant.}
\end{align*}
}
\end{proof}

\begin{lemma}
\label{r-prop1}
W.h.p. at the end of round~$i$, every dense vertex~$v$ satisfies the bounds
$$
a_i(v) \leq  D',\quad \extd_i(v) \leq D',\quad Q_i(v)  \geq Z'
$$
for the parameters
$$
D' = 12 D \sqrt{\delta} \qquad Z' = D/\sqrt{\delta}
$$
\end{lemma}
\begin{proof}
Let $v \in C_j$. We first note that the regularity conditions imply $D \delta \geq K \ln n$ and hence $D \geq \frac{K \ln n}{\sqrt{\delta}}$.

So we may apply Lemma~\ref{prop3} with $T$ being the set of external neighbors of~$v$ and $s = D$ to show that that $\extd_i(v) \leq D'$ holds with probability at least~$1 - n^{-K'}$ for an arbitrarily large constant $K'>0$. Similarly, we apply Lemma~\ref{prop3} with $T = C_j \setminus N(v)$ and $s = D$ to show that that $a_i(v) \leq D'$ holds with probability~$\geq 1 - n^{-K'}$.  

Next, we bound~$Q_i(v)$. We color at most~$D$ external neighbors and at most~$L_j$ internal neighbors. Thus, the residual palette of~$v$ has size at least~$Q_i(v) - L_j - D$. By Lemma~\ref{aprop1}, this is at least~$D/\sqrt{\delta}$ for $K$ sufficiently large.

Finally, take a union bound over all dense vertices $v$.
\end{proof}

\begin{proposition}
\label{aprop2}
W.h.p. at the end of round~$i$, every almost-clique $C_j$ has size at most
$$
\max(12 K\ln n, 12 M_j \sqrt{\delta} )
$$
\end{proposition}
\begin{proof}
 Apply Lemma~\ref{prop3} with $T = C_j$ and $s = \max( M_j, \frac{K \ln n}{\sqrt{\delta}})$; this shows that with probability at least~$1 - n^{-K'}$ for an arbitrarily large constant $K'>0$  there are at most $\max(12 K \ln n, 12 M_j \sqrt{\delta})$ uncolored vertices remaining in~$C_j$. Finally, take a union bound over all almost-cliques~$C_j$.
\end{proof}

\section{Solving the recurrence}
\label{sec:solver}
In light of Lemma~\ref{r-prop1}, we can explicitly derive a recurrence relation for the parameters~$D_i$ and~$Z_i$. We define $\delta_i = D_i/Z_i$ throughout.

\begin{lemma}
\label{r-prop0}
Define the recurrence relation with initial conditions
$$
D_0 = 3 \epsilon \Delta \qquad Z_0 = \Delta/2
$$
and recurrence
$$
D_{i+1} = 12 D_i \sqrt{\delta_i} \qquad Z_{i+1} = D_i/\sqrt{\delta_i}
$$

Let $i \leq \lceil \sqrt{\ln \Delta} \rceil$. Assuming that the regularity conditions (R1), (R2) are satisfied for $j = 0, \dots, i-1$, then we have w.h.p.:
$$
a_i(v) \leq D_i, \extd_i(v) \leq D_i, Q_i(v) \geq Z_i
$$
\end{lemma}
\begin{proof}
The bound on $Z_0$ is given in Theorem~\ref{initial-coloring-good-prop}. By Lemma~\ref{network1} and Lemma~\ref{network2}, we get~$a(v) \leq 3 \epsilon \Delta$ and  $\extd(v) \leq 3 \epsilon \Delta$ in the initial graph. The initial coloring step cannot increase these parameters, so we have $a_0(v) \leq 3 \epsilon \Delta, \extd_0(v) \leq 3 \epsilon \Delta$ as well. This shows the bound on~$D_0$.

A simple induction, using Lemma~\ref{r-prop1}, shows that for all $i = 1, \dots, n$ we have the following:
$$
a_i(v) \leq D_i, \extd_i(v) \leq D_i, Q_i(v)  \geq Z_i \qquad \text{with probability $1 - n^{-K'}$}
$$

Thus, for any fixed $i \leq \lceil \sqrt{\ln \Delta} \rceil$, the probability that any of these events does not occur, is at most $(1 + \sqrt{\ln \Delta}) n^{-K'} \leq n^{-K'+1}$ for an arbitrarily large constant $K'>0$.
\end{proof}

We will now show how to solve this recurrence.
\begin{lemma}
\label{r-prop3}
For all $i \leq \lceil \sqrt{\ln \Delta} \rceil$ we have $\delta_i = 6 \epsilon \cdot 12^i \leq 1/K$.
\end{lemma}
\begin{proof}
For each~$i > 0$, we may compute $\delta_i$ as
\begin{align*}
\delta_i &= \frac{D_i}{Z_i} = \frac{12 D_{i-1} \sqrt{\delta_{i-1}}}{D_{i-1}/\sqrt{\delta_{i-1}}} = 12 \delta_{i-1}
\end{align*}

As  $\delta_0 = D_0/Z_0 = 6 \epsilon$, we have $\delta_i = 6 \epsilon \cdot 12^i$ as claimed. Now, recalling our formula $\epsilon = 100^{-\sqrt{\ln \Delta}} / (100 K)$, we have 
$$
\delta_i \leq 6 \bigl( \frac{ 100^{-\sqrt{\ln \Delta}}}{100 K}  \bigr) \cdot  12^{\sqrt{\ln \Delta}+1} \leq \frac{ 72 }{100 K} 100^{-\sqrt{\ln \Delta}} 12^{\sqrt{\ln \Delta}} \leq 1/K.
$$
\end{proof}

\begin{corollary}
\label{gamma-cor}
For every $i \leq \lceil \sqrt{\ln \Delta} \rceil$, we have $\gamma_i \in [0,1]$.
\end{corollary}
\begin{proof}
Follows immediately from Lemma~\ref{r-prop3} and the definition $\gamma_i = 1 - 2 \sqrt{\delta_{i-1}}$.
\end{proof}

\begin{lemma}
\label{r-prop2}
For all $5 \leq i \leq \lceil \sqrt{\ln \Delta} \rceil$, we have $D_i \leq 12^{i^2/2} \cdot 10^{-i \sqrt{\ln \Delta}} \cdot \Delta$.
\end{lemma}
\begin{proof}
We recursively compute $D_i$ from $D_0$ as:
\begin{equation}
\label{d-rec-eqn}
D_i = D_0 \cdot \prod_{j=0}^{i-1} 12 \sqrt{\delta_j}
\end{equation}
Using Lemma~\ref{r-prop3} we then estimate:
\begin{align*}
D_i &\leq 3 \epsilon \Delta \prod_{j=0}^{i-1}  12 \delta_j^{1/2} \leq 3 \epsilon \Delta \prod_{j=0}^{i-1} \Bigl( 12 (6 \epsilon \cdot 12^j)^{1/2} \Bigr) = (3\epsilon\Delta)\cdot \Bigl( 12^{i} (6\epsilon)^{i/2} \cdot 12^{i(i-1)/4} \Bigr) \\
&\leq \Delta \cdot \Bigl (12^{i} \bigl( \frac{6 \cdot 100^{- \sqrt{\ln \Delta}}}{100 K} \bigr)^{i/2} \cdot 12^{i(i-1)/4} \Bigr) \leq \Delta \cdot 12^{i(i+5)/4} \cdot 100^{-i \sqrt{\ln\Delta} /2} \\
&\leq \Delta \cdot 12^{i^2/2} \cdot 10^{-i \sqrt{\ln\Delta}}  \qquad \text{as $(i+5)/4 \leq i/2$ for $i\geq 5$}
\end{align*}
\end{proof}

\begin{corollary}
\label{cor-dbound}
We have the bound $D_{\lceil \sqrt{\ln \Delta} \rceil} = O(1)$.
\end{corollary}
\begin{proof}
If $\sqrt{\ln \Delta} \leq 5$, then $D_i \leq D_0 = 3 \epsilon \Delta = O(1)$. Otherwise, we apply Lemma~\ref{r-prop2}:
\begin{align*}
D_{\lceil \sqrt{\ln \Delta} \rceil} &\leq  \Delta \cdot 12^{\lceil \sqrt{\ln\Delta} \rceil^2 /2} \cdot 10^{- \sqrt{\ln\Delta} \lceil{\sqrt{\ln \Delta} \rceil}} \leq  \Delta \cdot (\sqrt{12})^{(\sqrt{\ln\Delta} + 1)^2} \cdot 10^{-\ln\Delta} \ = O(1)
\end{align*}
\end{proof}

\iffalse
\textbf{Expository remark:} Corollary~\ref{cor-dbound} explains why we selected $\epsilon = \exp(-\Theta(\sqrt{\log \Delta}))$ and ran our coloring steps for $O(\sqrt{\log \Delta})$ rounds. Suppose instead we set~$\epsilon = \exp(-\ln^a \Delta)$ and ran $\ln^a \Delta$ dense coloring steps, for some~$a < 1/2$. At the end of these steps, we would have $D_{\ln^a \Delta} = \Delta \exp(-\ln^{2 a} \Delta)$. This is close to $\Delta$ (differing in only a sub-polynomial term), which implies that we have hardly made any progress in reducing the number of uncolored vertices.
\fi

\begin{theorem}
\label{th1}
At the end of the dense coloring steps, w.h.p. every dense vertex is connected to at most $O(\log n)\cdot 2^{O(\sqrt{\log \Delta})}$ other dense vertices.
\end{theorem}
\begin{proof}
Let $i^* \leq \lceil \sqrt{\ln \Delta} \rceil$ be minimal such that $D_{i^*} \delta_{i^*} \leq K \ln n$; Corollary~\ref{cor-dbound} ensures such an $i^*$ exists. 

The regularity conditions are satisfied up to round~$i^*$. Noting that $\delta_{i^*} \geq \epsilon = 2^{-\Theta(\sqrt{\log \Delta})}$, Lemma~\ref{r-prop0} shows that:
\begin{equation}\label{eqn:Dbound}
\extd_{i^*}(v) \leq D_{i^*} \leq (K \ln n)/\delta_{i^{*}} = O( \log n ) \cdot 2^{O(\sqrt{\log \Delta})}
\end{equation}

Next, we bound the size of each almost-clique $C_j$. The initial size of $C_j$ is at most $(1+3 \epsilon) \Delta$. Applying Proposition~\ref{aprop2} repeatedly for $i < i^*$ shows that the size of $C_j$ reduces to 
$$
\max(12 K\ln n, (1+3\epsilon) \Delta \cdot \prod_{i=0}^{i^*-1} 12 \sqrt{\delta_i})
$$
w.h.p. We bound the latter term as follows:
{\allowdisplaybreaks
\begin{align*}
(1+3\epsilon) \Delta \cdot \prod_{i=0}^{i^*-1} 12 \sqrt{\delta_j} &= (1+3\epsilon) \Delta \cdot \frac{ D_{i*}}{3 \epsilon \Delta} \qquad \text{by (\ref{d-rec-eqn}) and $D_0 = 3\epsilon\Delta$} \\
 &\leq \frac{(1+3\epsilon)K \ln n}{3 \epsilon} \cdot 2^{O(\sqrt{\log \Delta})} \qquad \text{by (\ref{eqn:Dbound})} \\
&= O(\log n) \cdot 2^{O(\sqrt{\log \Delta})} \qquad \text{as $1/\epsilon  2^{\Theta(\sqrt{\log \Delta})}$}
\end{align*}
}

We have shown that  $v \in C_j$ has $O(\log n)\cdot 2^{O(\sqrt{\log \Delta})}$ external neighbors and  $O(\log n)\cdot 2^{O(\sqrt{\log \Delta})}$ neighbors $w \in C_j$ after round $i^*$. Dense coloring steps after round $i^*$ can only decrease the degree of $v$, so $v$ has $O(\log n)\cdot 2^{O(\sqrt{\log \Delta})}$ dense neighbors after round $\lceil \sqrt{\ln \Delta} \rceil$. 
\end{proof}

\iffalse
We have shown that after the $\lceil \sqrt{\ln \Delta} \rceil$ dense coloring steps, the number of dense neighbors of each dense vertex shrinks to $O(\log n)\cdot 2^{O(\sqrt{\log \Delta})}$. Also, for each sparse vertex~$x$, we have $Q_{\lceil \sqrt{\ln \Delta} \rceil}(x) \geq \deg_{\lceil \sqrt{\ln \Delta} \rceil}(x) + \Omega(\epsilon^2 \Delta)$ due to the initial coloring step. By applying the algorithm of Elkin, Pettie, and Su~\cite[Section 4]{elk15} on the sparse component, it can be colored in $O(\log(1/\epsilon)) + 2^{O(\sqrt{\log \log n})} = O(\sqrt{\log \Delta}) + 2^{O(\sqrt{\log \log n})}$ rounds. Then, we apply the algorithm of Barenboim et al.~\cite{BEPS16} to color the remaining vertices whose degrees are bounded by $\Delta' = O(\log n)\cdot 2^{O(\sqrt{\log \Delta})}$. It then runs in $O(\log \Delta') + 2^{O(\sqrt{\log \log n})} = O(\sqrt{\log \Delta}) +  2^{O(\sqrt{\log \log n})}$ rounds. The total number of rounds is $O(\sqrt{\log \Delta}) + 2^{O(\sqrt{\log \log n})}$.
\fi

\section{List-coloring locally-sparse graphs}
\label{list-color-locally-sparse-sec}
Although the overall focus of this paper is an algorithm for coloring arbitrary graphs in time $O(\sqrt{\log \Delta}) + 2^{O(\sqrt{\log \log n})}$, we note that our initial coloring step may also be used to obtain a faster list-coloring algorithm for sparse graphs. This result extends the work of~\cite{elk15}, which showed a similar type of $(\Delta+1)$-coloring algorithm for graphs which satisfy a property they refer to as \emph{local sparsity.} We define this property and show that it is essentially equivalent to the definition of sparsity defined in Section~\ref{decomp-sec}.
\begin{definition}
A graph~$G$ is \emph{$(1-\delta)$-locally sparse} if very vertex contains at most $(1-\delta) \binom{\Delta}{2}$ edges in its neighborhood, for some parameter $\delta \in [0,1]$. (That is, the induced subgraph $G[ N(v) ]$ contains $\leq (1 - \delta) \binom{\Delta}{2}$ edges).
\end{definition}

\iffalse
\begin{lemma}
Suppose that we have a network decomposition with sparsity parameter~$\epsilon$. The neighborhood of a sparse vertex~$u$ spans at most $(1-\epsilon^2)\Delta^2/2$ edges. That is, $|\{xy \in E \mid x,y \in N(u) \}| \leq (1-\epsilon^2)\Delta^2/2$.
 \end{lemma}
\begin{proof}
We can count the number of edges in the neighborhood of $u$ by:
\begin{align*}
|\{xy \mid x,y \in N(u) \}| &= \frac{1}{2} \sum_{x \in N(u)}  |N(u) \cap N(x)| \qquad \text{(double counting)} \\
&= \frac{1}{2} \sum_{\substack{x \in N(u)\\ ux \in F}}  |N(u) \cap N(x)| + \frac{1}{2} \sum_{\substack{x \in N(u)\\ ux \notin F}}  |N(u) \cap N(x)| \\
&\leq \frac{1}{2} \sum_{\substack{x \in N(u)\\ ux \in F}}  \Delta + \frac{1}{2} \sum_{\substack{x \in N(u)\\ ux \notin F}}  (1 - \epsilon) \Delta \\
&= \frac{1}{2} \sum_{\substack{x \in N(u)\\ ux \in F}} \epsilon \Delta + \frac{1}{2} \sum_{x \in N(u)}  (1 - \epsilon) \Delta \\
&\leq \frac{1}{2} (1 - \epsilon) \Delta \epsilon \Delta + \frac{1}{2} \Delta  (1 - \epsilon) \Delta \qquad \text{($u$ has at most $(1-\epsilon) \Delta$ friends)} \\
&= \Delta^2 (1 - \epsilon^2) / 2
\end{align*}
\end{proof}
\fi

\begin{lemma}\label{sparse-equiv-prop}
Suppose that $G$ is $(1-\delta)$-locally sparse. If we apply the network decomposition of Section~\ref{decomp-sec} with parameter $\epsilon = \delta/2$, then every vertex is sparse, i.e. $V^{\text{sparse}} = V$.
\end{lemma}
\begin{proof}
Suppose that $v \in V$ is dense with respect to~$\epsilon$. Then $v$ has at least~$(1 - \epsilon) \Delta$ friends. Each such friend~$u$ corresponds to at least~$(1-\epsilon) \Delta$ edges between $u$ and another neighbor of~$v$, that is, an edge in $G[ N(v) ]$. Furthermore, any such edge in $G[ N(v) ]$ is counted at most twice, so $G[ N(v) ]$ must contain $(1-\epsilon)^2 \Delta^2/2 \geq (1 - 2 \epsilon) \binom{\Delta}{2}$ edges, which contradicts our hypothesis for $\epsilon \geq \delta/2$.
\end{proof}

\begin{corollary}
Suppose that $G$ is $(1-\delta)$-locally-sparse and that every vertex has a palette of size exactly~$\Delta+1$. Then $G$ can be list-colored w.h.p. in $O( \log(1/\delta)) + 2^{O(\sqrt{\log \log n})}$ rounds.
\end{corollary}
\begin{proof}
By Proposition~\ref{sparse-equiv-prop}, every vertex in $G$ is sparse with respect to parameter $\epsilon = \delta/2$. 

First suppose that $\delta^4 \Delta \geq K \ln n$, where $K$ is a sufficiently large constant. Then by Theorem~\ref{initial-coloring-good-prop}, each vertex satisfies $S_0(v) = \Omega(\epsilon^2 \Delta)$ w.h.p. The algorithm of~\cite{elk15} applied to the residual graph runs in $O(\log(1/\delta^2)) + 2^{O(\sqrt{\log \log n})}$ rounds.

Next, suppose that $\delta^4 \Delta \leq K \ln n$. Then the coloring algorithm of~\cite{BEPS16} runs in $O(\log \Delta) + 2^{O(\sqrt{\log \log n})} = O(\log(1/\delta)) + 2^{O(\sqrt{\log \log n})}$ rounds.
\end{proof}

\section{Conclusions}
Distributed symmetry breaking tasks such as coloring or MIS lie at the heart of distributed computing. We have shown that the $(\Delta+1)$-coloring problem is easier than MIS. However, there is still a large gap in the round complexity between the best lower bound of $\Omega(\log^* n)$ and our upper bound of $O(\sqrt{\log \Delta})+ 2^{O(\sqrt{\log \log n})}$. Recent advances for lower bounds~\cite{Brandt16, ChangP17, hef16, chlpj17} for the \textsf{LOCAL} model and its variants might yield inspiration for improving the existing lower bound. Our deterministic decomposition into locally sparse and dense parts might foster additional advances as well. It might help to further reduce upper bounds for symmetry breaking tasks -- in particular for deterministic algorithms, since there exist efficient algorithms for (very) dense graphs, e.g. ~\cite{sch10opt}, and sparse graphs, e.g.~\cite{moh17}. Furthermore, the gap for $(\Delta+1)$-coloring between our randomized algorithm running in time~$O(\sqrt{\log \Delta})+ 2^{O(\sqrt{\log \log n})}$ and the best deterministic algorithm requiring $O(\sqrt{\Delta}\log^{2.5}\Delta + \log^{*} n)$~\cite{FraigniaudHK15} is more than exponential and, therefore, larger than any known separation result for randomized and deterministic algorithms~\cite{chan16}.

\medskip
\noindent \textbf{Acknowledgments:} We would like to thank Seth Pettie for his valuable comments. We thank the anonymous reviewers for many helpful suggestions and comments.

\bibliographystyle{abbrv}
\bibliography{coloralg-references}

\begin{thebibliography}{10}

\bibitem{alon86}
N.~Alon, L.~Babai, and A.~Itai.
\newblock A fast and simple randomized parallel algorithm for the maximal
  independent set problem.
\newblock {\em Journal of Algorithms}, 7(4):567--583, 1986.

\bibitem{awer89}
B.~Awerbuch, A.~V. Goldberg, M.~Luby, and S.~A. Plotkin.
\newblock {Network decomposition and locality in distributed computation}.
\newblock In {\em Proc. of Symp. Foundations of Computer Science (FOCS)}, 1989.

\bibitem{bar15a}
L.~Barenboim.
\newblock Deterministic {$(\Delta+ 1)$}-coloring in sublinear (in {$\Delta$})
  time in static, dynamic and faulty networks.
\newblock In {\em Symp. on Principles of Distributed Computing(PODC)}, pages
  345--354, 2015.

\bibitem{elk10}
L.~Barenboim and M.~Elkin.
\newblock Deterministic distributed vertex coloring in polylogarithmic time.
\newblock {\em Journal of the ACM (JACM)}, 58(5):23, 2011.

\bibitem{bar13}
L.~Barenboim and M.~Elkin.
\newblock Distributed graph coloring: Fundamentals and recent developments.
\newblock {\em Synthesis Lectures on Distributed Computing Theory},
  4(1):1--171, 2013.

\bibitem{bar15}
L.~Barenboim, M.~Elkin, and C.~Gavoille.
\newblock A fast network-decomposition algorithm and its applications to
  constant-time distributed computation.
\newblock {\em Journal of Theoretical Computer Science}, 2016.

\bibitem{BEK09}
L.~Barenboim, M.~Elkin, and F.~Kuhn.
\newblock Distributed {$(\Delta+1)$}-coloring in linear (in ${\Delta}$) time.
\newblock {\em SIAM Journal on Computing}, 43(1):72--95, 2014.

\bibitem{BEPS16}
L.~Barenboim, M.~Elkin, S.~Pettie, and J.~Schneider.
\newblock The locality of distributed symmetry breaking.
\newblock {\em Journal of the ACM (JACM)}, 63, 2016.
\newblock Article 20.

\bibitem{Brandt16}
S.~Brandt, O.~Fischer, J.~Hirvonen, B.~Keller, T.~Lempi\"{a}inen, J.~Rybicki,
  J.~Suomela, and J.~Uitto.
\newblock A lower bound for the distributed {L}ov\'{a}sz local lemma.
\newblock In {\em Proc. of Symposium on Theory of Computing (STOC)}, pages
  479--488, 2016.

\bibitem{chlpj17}
Y.-J. Chang, Q.~He, W.~Li, S.~Pettie, and J.~Uitto.
\newblock The complexity of distributed edge coloring with small palettes.
\newblock In {\em arXiv preprint arXiv:1708.04290}, 2017.

\bibitem{chan16}
Y.-J. Chang, T.~Kopelowitz, and S.~Pettie.
\newblock An exponential separation between randomized and deterministic
  complexity in the local model.
\newblock In {\em Proc. of Foundations of Computer Science (FOCS)}, pages
  615--624. IEEE, 2016.

\bibitem{ChangP17}
Y.-J. Chang and S.~Pettie.
\newblock A time hierarchy theorem for the {LOCAL} model.
\newblock In {\em Proc. of Symp. on Foundations of Computer Science (FOCS)},
  2017.

\bibitem{CK85}
I.~Chlamtac and S.~Kutten.
\newblock A spatial reuse {TDMA/FDMA} for mobile multi-hop radio networks.
\newblock In {\em Proc. of the Joint Conference of the IEEE Computer and
  Communications Societies (INFOCOM). Technology: Emerging or Converging},
  pages 389--394 vol.1, 1985.

\bibitem{CPS17}
K.-M. Chung, S.~Pettie, and H.-H. Su.
\newblock Distributed algorithms for the {L}ov{\'a}sz local lemma and graph
  coloring.
\newblock {\em Journal of Distributed Computing}, 30(4):261--280, 2017.

\bibitem{CS89}
I.~Cidon and M.~Sidi.
\newblock Distributed assignment algorithms for multihop packet radio networks.
\newblock {\em IEEE Transactions on Computers}, 38(10):1353--1361, 1989.

\bibitem{DGP98}
D.~Dubhashi, D.~A. Grable, and A.~Panconesi.
\newblock Near-optimal, distributed edge colouring via the nibble method.
\newblock {\em Journal of Theoretical Computer Science (TCS)}, 203(2):225--251,
  Aug. 1998.

\bibitem{elk15}
M.~Elkin, S.~Pettie, and H.-H. Su.
\newblock {$(2 \Delta- 1)$}-edge-coloring is much easier than maximal matching
  in the distributed setting.
\newblock In {\em Symp. on Discrete Algorithms (SODA)}, pages 355--370. SIAM,
  2015.

\bibitem{ET90}
A.~Ephremides and T.~V. Truong.
\newblock Scheduling broadcasts in multihop radio networks.
\newblock {\em IEEE Transactions on Communications}, 38(4):456--460, 1990.

\bibitem{fischer2017sublogarithmic}
M.~Fischer and M.~Ghaffari.
\newblock Sublogarithmic distributed algorithms for {L}ov{\'a}sz local lemma
  with implications on complexity hierarchies.
\newblock In {\em Proceedings 31st International Symposium on Distributed
  Computing (DISC)}, page~18, 2017.

\bibitem{FraigniaudHK15}
P.~Fraigniaud, M.~Heinrich, and A.~Kosowski.
\newblock Local conflict coloring.
\newblock In {\em Proc. of Foundations of Computer Science (FOCS)}, pages
  625--634. IEEE, 2016.

\bibitem{Moh16}
M.~Ghaffari.
\newblock An improved distributed algorithm for maximal independent set.
\newblock In {\em Proc. 27th Annual ACM-SIAM Symposium on Discrete Algorithms
  (SODA)}, pages 270--277, 2016.

\bibitem{moh17}
M.~Ghaffari and C.~Lymouri.
\newblock Simple and near-optimal distributed coloring for sparse graphs.
\newblock In {\em arXiv preprint arXiv:1708.06275}, 2017.

\bibitem{GP87}
A.~V. Goldberg and S.~A. Plotkin.
\newblock Parallel {$(\Delta+1)$}-coloring of constant-degree graphs.
\newblock {\em Information Processing Letters}, 25(4):241 -- 245, 1987.

\bibitem{GPS88}
A.~V. Goldberg, S.~A. Plotkin, and G.~E. Shannon.
\newblock Parallel symmetry-breaking in sparse graphs.
\newblock {\em {SIAM} J. Discrete Math.}, 1(4):434--446, 1988.

\bibitem{goo14}
M.~G\"{o}\"{o}s, J.~Hirvonen, and J.~Suomela.
\newblock Linear-in-{$\Delta$} lower bounds in the local model.
\newblock In {\em Proc. of Symposium on Principles of Distributed
  Computing(PODC)}, pages 86--95, 2014.

\bibitem{GP97}
D.~A. Grable and A.~Panconesi.
\newblock Nearly optimal distributed edge coloring in ${O}(\log \log n)$
  rounds.
\newblock {\em Journal of Random Structures \& Algorithms}, 10(3):385--405,
  1997.

\bibitem{GP00}
D.~A. Grable and A.~Panconesi.
\newblock Fast distributed algorithms for {B}rooks--{V}izing colorings.
\newblock {\em Journal of Algorithms}, 37(1):85 -- 120, 2000.

\bibitem{hef16}
D.~Hefetz, F.~Kuhn, Y.~Maus, and A.~Steger.
\newblock Polynomial lower bound for distributed graph coloring in a weak local
  model.
\newblock In {\em Proc. of International Symposium on Distributed
  Computing(DISC)}, pages 99--113. Springer, 2016.

\bibitem{hir12}
J.~Hirvonen and J.~Suomela.
\newblock Distributed maximal matching: Greedy is optimal.
\newblock In {\em Proc. of Symp. on Principles of Distributed Computing(PODC)},
  pages 165--174, 2012.

\bibitem{Joh99}
{\"O}.~Johansson.
\newblock Simple distributed {$\Delta+1$}-coloring of graphs.
\newblock {\em Inf. Process. Lett.}, 70, 1999.

\bibitem{KSOS06}
K.~Kothapalli, C.~Scheideler, M.~Onus, and C.~Schindelhauer.
\newblock Distributed coloring in $\tilde{O}(\sqrt{\log n})$ bit rounds.
\newblock In {\em Proc. of International Parallel and Distributed Processing
  Symposium {(IPDPS)}}, 2006.

\bibitem{kuh10}
F.~Kuhn, T.~Moscibroda, and R.~Wattenhofer.
\newblock Local computation: Lower and upper bounds.
\newblock {\em Journal of the ACM (JACM)}, 63(2):17, 2016.

\bibitem{Kuhn2006On}
F.~Kuhn and R.~Wattenhofer.
\newblock On the complexity of distributed graph coloring.
\newblock In {\em { Symp. on Principles of Distributed Computing (PODC)}},
  2006.

\bibitem{linial92}
N.~Linial.
\newblock {Locality in Distributed Graph Algorithms}.
\newblock {\em SIAM Journal on Computing}, 21(1):193--201, 1992.

\bibitem{lin93}
N.~Linial and M.~Saks.
\newblock Low diameter graph decompositions.
\newblock {\em Combinatorica}, 13(4):441--454, 1993.

\bibitem{lov79}
L.~Lov{\'a}sz.
\newblock {\em Combinatorial problems and exercises}.
\newblock North Holland, 1979.

\bibitem{lub86}
M.~Luby.
\newblock A simple parallel algorithm for the maximal independent set problem.
\newblock {\em SIAM Journal on Computing}, 15:1036--1053, 1986.

\bibitem{mol98}
M.~Molloy and B.~Reed.
\newblock A bound on the total chromatic number.
\newblock {\em Combinatorica}, 18(2):241--280, 1998.

\bibitem{MR01}
M.~Molloy and B.~Reed.
\newblock {\em Graph Colouring and the Probabilistic Method}.
\newblock Algorithms and Combinatorics. Springer, 2001.

\bibitem{mol10}
M.~Molloy and B.~Reed.
\newblock Asymptotically optimal frugal colouring.
\newblock {\em Journal of Combinatorial Theory, Ser. {B}}, 100(2):226--246,
  2010.

\bibitem{mol14}
M.~Molloy and B.~Reed.
\newblock Colouring graphs when the number of colours is almost the maximum
  degree.
\newblock {\em Journal of Combinatorial Theory, Series B}, 109:134 -- 195,
  2014.

\bibitem{Noar91}
M.~Naor.
\newblock A lower bound on probabilistic algorithms for distributive ring
  coloring.
\newblock {\em SIAM Journal on Discrete Mathematics}, 4(3):409--412, 1991.

\bibitem{panc92}
A.~Panconesi and A.~Srinivasan.
\newblock Improved distributed algorithms for coloring and network
  decomposition problems.
\newblock In {\em Proc. of Symposium on Theory of Computing(STOC)}, pages
  581--592. ACM, 1992.

\bibitem{PS97}
A.~Panconesi and A.~Srinivasan.
\newblock Randomized distributed edge coloring via an extension of the
  {C}hernoff--{H}oeffding bounds.
\newblock {\em SIAM Journal on Computing}, 26(2):350--368, 1997.

\bibitem{PS13}
S.~Pettie and H.-H. Su.
\newblock Distributed coloring algorithms for triangle-free graphs.
\newblock {\em Information and Computation}, 243:263--280, 2015.

\bibitem{RP89}
R.~Ramaswami and K.~K. Parhi.
\newblock Distributed scheduling of broadcasts in a radio network.
\newblock In {\em Proc. of the Conference of the IEEE Computer and
  Communications Societies (INFOCOM). Technology: Emerging or Converging,
  IEEE}, pages 497--504 vol.2, 1989.

\bibitem{ree98}
B.~Reed.
\newblock $\omega$, {$\Delta$}, and $\chi$.
\newblock {\em Journal of Graph Theory}, 27(4):177--212, 1998.

\bibitem{ree99}
B.~Reed.
\newblock A strengthening of {Brooks'} theorem.
\newblock {\em Journal of Combinatorial Theory, Series B}, 76(2):136 -- 149,
  1999.

\bibitem{sch13}
J.~Schneider, M.~Elkin, and R.~Wattenhofer.
\newblock Symmetry breaking depending on the chromatic number or the
  neighborhood growth.
\newblock {\em Journal of Theoretical Computer Science(TCS)}, 509:40--50, 2013.

\bibitem{Sch10}
J.~Schneider and R.~Wattenhofer.
\newblock A new technique for distributed symmetry breaking.
\newblock In {\em {Symp. on Principles of Distributed Computing(PODC)}}, 2010.

\bibitem{sch10opt}
J.~Schneider and R.~Wattenhofer.
\newblock An optimal maximal independent set algorithm for bounded-independence
  graphs.
\newblock {\em Journal of Distributed Computing}, 22(5):349--361, 2010.

\bibitem{SV93}
M.~Szegedy and S.~Vishwanathan.
\newblock Locality based graph coloring.
\newblock In {\em Proc. of Symposium on Theory of Computing (STOC)}, pages
  201--207, 1993.

\end{thebibliography}
\end{document}